\documentclass[journal,10pt,twocolumn,twoside]{IEEEtran}

\usepackage{cite}

\ifCLASSINFOpdf
   \usepackage[pdftex]{graphicx}
\else

\fi

\usepackage[cmex10]{amsmath}
\usepackage{breqn}
\usepackage{algorithm}
\usepackage{algcompatible}
\usepackage{array}
\usepackage[caption=false,font=normalsize,labelfont=sf,textfont=sf]{subfig}

\usepackage{url}
\usepackage{scalerel}

\usepackage{xcolor}

\definecolor{Darkgray}{gray}{0}

\usepackage{multicol,multirow}
\usepackage{cleveref}

\usepackage{amssymb}

\usepackage{amsthm}
\usepackage{makecell}
\usepackage{mathtools}

\usepackage[]{footmisc}

\newtheorem{proposition}{\bf Proposition}

\newtheorem{remark}{Remark}

\newcommand{\ve}[1]{\boldsymbol{#1}}

\newcommand{\argmax}{\operatornamewithlimits{argmax}}
\newcommand{\argmin}{\operatornamewithlimits{argmin}}
\newcommand\abs[1]{\left|#1\right|}

\newcolumntype{I}{!{\vrule width 1.2pt}}
\makeatletter
\def\hlinewd#1{%
\noalign{\ifnum0=`}\fi\hrule \@height #1 %
\futurelet\reserved@a\@xhline}

\makeatletter 
\let\myorg@bibitem\bibitem
\def\bibitem#1#2\par{%
  \@ifundefined{bibitem@#1}{%
    \myorg@bibitem{#1}#2\par
  }{%
    \begingroup
      \color{\csname bibitem@#1\endcsname}%
      \myorg@bibitem{#1}#2\par
    \endgroup
  }%
}
\newcommand*{\bibitem@basardynamic}{black}    
\makeatother 

\makeatletter
\newcommand*\titleheader[1]{\gdef\@titleheader{#1}}
\AtBeginDocument{%
  \let\st@red@title\@title
  \def\@title{%
    \bgroup\normalfont\large\centering\@titleheader\par\egroup
    \vskip1.5em\st@red@title}
}
\makeatother

\begin{document}

\title{Network-Aware Demand-side Management Framework with A Community Energy Storage System Considering Voltage Constraints}
\author{Chathurika~P.~Mediwaththe,~\IEEEmembership{Member,~IEEE,}~Lachlan Blackhall,~\IEEEmembership{Senior~Member,~IEEE} %
\thanks{This work was supported by the Advancing Renewables Program, Australian Renewable Energy Agency under Grant 2018/ARP134.}
\thanks{ C. P. Mediwaththe and L. Blackhall are with the College of Engineering and Computer Science, The Australian National University, Canberra, ACT 0200, Australia.
e-mail: (chathurika.mediwaththe@anu.edu.au, lachlan.blackhall@anu.edu.au).}}%

\markboth{Accepted to appear in IEEE Transactions on Power Systems, 2020}%
{Accepted to appear in IEEE Transactions on Power Systems, 2020}
%

\maketitle

\begin{abstract}
This paper studies the feasibility of integrating a community energy storage (CES) system with rooftop photovoltaic (PV) power generation for demand-side management of a neighbourhood while maintaining the distribution network voltages within allowed limits. To this end, we develop a decentralized energy trading system between a CES provider and users with rooftop PV systems. By leveraging a linearized branch flow model for radial distribution networks, a voltage-constrained leader-follower Stackelberg game is developed wherein the CES provider maximizes revenue and the users minimize their personal energy costs by trading energy with the CES system and the grid. The Stackelberg game has a unique equilibrium at which the CES provider maximizes revenue and the users minimize energy costs at a unique Nash equilibrium. A case study, with realistic PV power generation and demand data, confirms that the energy trading system can reduce peak energy demand and prevent network voltage excursions, while delivering financial benefits to the users and the CES provider. Further, simulations highlight that, in comparison with a centralized system, the decentralized energy trading system provides greater economic benefits to the users with less energy storage capacity.
\end{abstract}

\begin{IEEEkeywords}
Community energy storage,~demand-side management,~distribution network,~ game theory,~photovoltaic power generation,~power flow,~voltage regulation.
\end{IEEEkeywords}

\IEEEpeerreviewmaketitle

\section*{Nomenclature}
\addcontentsline{toc}{section}{Nomenclature}
\begin{IEEEdescription}[\IEEEusemathlabelsep\IEEEsetlabelwidth{$V_1,V_2~~~~~~~~~~~$}] 
\item[\textit{Sets and Indices}] ~~~~~\vspace{5pt}
\item[$\mathcal{A},~a$] Set of all users, user index. 
\item[$\mathcal{N},~n$] Set of non-participating users, user index.
\item[$\mathcal{P},~p$] Set of participating users, user index.
\item[$\mathcal{T},~t$] Set of all time intervals, time index.
\item[$\mathcal{V}$, $i, j$] Set of buses (nodes), bus indices. 
\item[$\mathcal{A}_i$] Set of all users at bus $i$.
\item[$\mathcal{N}_i$] Set of non-participating users at bus $i$.
\item[$\mathcal{P}_i$] Set of participating users at bus $i$.
\item[ $\mathcal{E}$]  Set of distribution lines (directed edges).
\item[$\mathcal{J}_i$] Set of edges on the unique path from bus 0 to $i$. 
\item[$\mathcal{U}$]  Feasible strategy set of the CES provider.
\item[$\mathcal{P}_{\scaleto{+}{4.5pt}}(t)$]  Set of surplus energy users at time $t$.
\item[$\mathcal{P}_{\scaleto{-}{4.5pt}}(t)$] Set of deficit energy users at time $t$. 
\item[$\mathcal{Y}_p(t)$] Feasible strategy set of user $p \in \mathcal{P}$ at time $t$.\vspace{5pt}

\item[\textit{Parameters and Constants}]  ~~~~~\vspace{5pt}

\item[$H$] Total number of time intervals in $\mathcal{T}$.
\item[$M$] Number of participating users, $|\mathcal{P}|=M$.
\item[$\theta$]  Small positive constant.
\item[$\eta_c,~\eta_d $] Charging and discharging efficiencies of the CES.
\item[$r_{ij},~x_{ij}$]  Resistance, reactance of line $(i,j) \in \mathcal{E}$.
\item[$\phi_t,~\delta_t $] External grid price constants at time $t$. 
\item[$B_{\text{max}},~B_{\text{min}}$] Maximum and minimum energy capacity limits of the CES.
\item[$E_{\text{rev,max}},~E_{\text{fw,max}}$]  Maximum reverse (rev) and forward (fw) energy flow limits on the external grid.
\item[$S_{ij,\text{max}}$] Apparent power rating of line $(i,~j)\in \mathcal{E}$.
\item[$V_{\text{max}},~V_{\text{min}}$] Maximum and minimum voltage magnitude limits of the distribution network.
\item[$\lambda_{\text{min}}$]  Constant lower price limit, $\lambda_{\text{min}} > 0.$
\item[$\gamma^{\text{ch}}_{\text{max}},~\gamma^{\text{dis}}_{\text{max}}$ ]  Maximum charging and discharging power rates of the CES.
\item[$d_{a}(t)$] Energy demand of user $a \in \mathcal{A}$ at time $t$. 
\item[$d_{p}(t)$] Energy demand of user $p \in \mathcal{P}$ at time $t$. 
\item[$d_{n}(t)$] Energy demand of user $n \in \mathcal{N}$ at time $t$.
\item[$g_p(t)$] PV energy generation of user $p \in \mathcal{P}$ at time $t$.
\item[$q_a(t)$] Reactive power demand of user $a \in \mathcal{A}$ at time $t$.
\item[$q_p(t)$] Reactive power demand of user $p \in \mathcal{P}$ at time $t$.\vspace{5pt}

\item[\textit{Functions and Variables}] ~~~~~\vspace{5pt}

\item[$P_i,~Q_i$] Active and reactive power consumptions of bus $i$. 
\item[$P_{ij},~Q_{ij}$] Active and reactive power flows from bus $i$ to $j$.  
\item[$V_i$] Voltage magnitude of bus $i$.
\item[$C_p(\cdot)$] Energy cost function of user $p \in \mathcal{P}$.
\item[$W_s (\cdot)$] Revenue function of the CES provider.
\item[$b(t)$]  CES charge level at the end of time $t$. 
\item[$e_g(t)$] Energy traded between the external grid and the CES at time $t$.
\item[$e_p(t)$] Energy that user $p \in \mathcal{P}$ trades with the external grid at time $t$.
\item[$e_s(t)$]  Actual energy flowed into/out of the CES at time $t$.
\item[$s_p(t)$] Surplus energy at user $p \in \mathcal{P}$ at time $t$. 
\item[$y_p(t)$] Energy that user $p \in \mathcal{P}$ trades with the CES at time $t$.
\item[$E(t)$] Total external grid energy at time $t$.
\item[$E_{\mathcal{P}}(t)$] Total external grid energy of the users $\mathcal{P}$ at time $t$.
\item[$E_{\mathcal{N}}(t)$] Total external grid energy of the users $\mathcal{N}$ at time $t$.
\item[$E_{-p}(t)$] Total external grid energy excluding external grid energy of user $p$ at time $t$. 
\item[$\lambda_g(t)$] External grid energy price at time $t$.
\item[$\lambda_s(t)$] CES provider's energy price at time $t$.\vspace{5pt}

\item[\textit{Other notations}]   ~~~~~\vspace{5pt}
\item[$\mathcal{G}$ ] Rooted tree in graph theory.
\item[$\Gamma $] Non-cooperative game of the users $\mathcal{P}$.
\item[$\Theta$] Stackelberg game between the CES provider and the users $\mathcal{P}$.
\item[$\mathcal{L}$] Leader.
\item[$\mathcal{F}$]  Followers.

\end{IEEEdescription}

\section{Introduction}\label{Intro}

The shift toward affordable electricity has increased the market penetration of rooftop photovoltaic (PV) systems worldwide. However, due to large-scale integration of rooftop PV systems, low-voltage distribution networks frequently experience over-voltage conditions when PV generation exceeds household demand. Distribution networks are also susceptible to voltage drops below permissible limits due to peak electricity demand. Energy storage systems (ESSs) located close to users have paved the way for developing effective voltage regulation methods for distribution networks and for developing demand-side management (DSM) methods to accomodate peak energy demand without upgrading the existing grid infrastructure \cite{feeder_model}. With the close proximity to users, community energy storage (CES) systems can be used to develop innovative DSM approaches by exploiting rooftop PV power generation. Additionally, CES systems can be utilized to mitigate voltage excursions in distribution networks without curtailing excessive PV power generation.

A DSM framework that utilizes a CES system with user-owned PV power generation requires economic incentives for all stakeholders including users and the storage provider. Decentralized methods that can distribute the decision-making to individuals are more preferable than centralized methods since obtaining full access to users’ personal energy usage information by a central entity in centralized DSM may be difficult and less practical \cite{Celik, MG3}. Additionally, scheduling a CES system by merely maximizing the economic benefits may not be feasible in practice since the energy dispatch schedules produced by such frameworks may violate critical network constraints such as voltage limits in distribution networks.

In this paper, we study the extent to which a CES system can mitigate voltage excursions in a distribution network with high penetration of rooftop PV systems while performing DSM. To this end, we develop a decentralized energy trading system between a CES provider and users with rooftop PV systems. In the system, users determine the energy amounts that can be traded with the CES system and the grid by minimizing personal energy costs. Additionally, the CES provider maximizes revenue by setting a price signal for the energy transactions with the users and by determining the energy transactions with the grid. The paper has the following key contributions: 
\begin{itemize}
\item By leveraging a linearized branch flow model, a voltage-constrained non-cooperative Stackelberg game is developed to study the decentralized energy trading between the CES provider and the users while complying with the voltage limits of a radial distribution network. 
\item We prove that the Stackelberg game has a unique pure strategy Stackelberg equilibrium at which the CES provider maximizes revenue and the users minimize personal energy costs at a unique Nash equilibrium.
\end{itemize} 

Non-cooperative game theory has been exploited to develop decentralized DSM frameworks that maximize economic benefits to individual users by coordinating ESSs while satisfying the network voltage constraints \cite{MG1,Gupta}. For instance, a DSM framework to maximize per-user economic benefits by coordinating user-owned ESSs has been studied in \cite{Gupta} by developing a voltage-constrained game among the users. In this context, to the best of our knowledge, this paper is the first to leverage Stackelberg game theory with a branch power flow model to study a decentralized energy trading framework between a CES system and users with PV generation for DSM in a neighborhood while satisfying the network voltage constraints.

This paper has a key difference to our previous work \cite{Chathu1}. In \cite{Chathu1}, the energy trading between the CES system and the users has been studied by merely maximizing the economic benefits for the storage provider and the users. In contrast, this work studies the physical network integration of the energy trading framework by explicitly taking into account the underlying distribution network voltage constraints. 

The remainder of the paper is organized as follows. Section~\ref{relWrk} presents related work, and Section~\ref{Sys_models} presents the system models of the energy trading system including the distribution network power flow model. Section~\ref{sec_system} presents the game-theoretic formulation of the system, and Section~\ref{sec:5} presents simulation results. Section~\ref{conclusion} concludes the paper.
 
\section{Related Work}\label{relWrk}

Scheduling ESSs to achieve DSM while complying with the distribution network constraints has generally been studied as an optimal power flow (OPF) problem. DSM-oriented OPF problems, that can be effectively optimized by a central entity, e.g., the distribution network operator, with the full system information availability, have been extensively studied in literature \cite{Tant, Karthikeyan2, Hu_Li, Gayme}. For instance, the OPF-based DSM method in \cite{Gayme} minimizes the total active power generation cost by scheduling ESSs subject to network voltage constraints. 
 
Decentralized methods have been explored by distributing the power scheduling computation of the centralized OPF problem to the local controllers at distributed ESSs and power generators \cite{MG3,DallAnese, Lam, Bitar}. For instance, in \cite{MG3}, a decentralized scheme based on predictor corrector proximal multiplier algorithm is proposed to solve a centralized OPF problem that can be used to find the cost-optimal active and reactive power set-points of ESSs in a microgrid. More recently, decentralized peer-to-peer energy trading frameworks have been proposed to maintain the demand-supply balance in distribution networks while ensuring the network voltages are within allowed limits \cite{1_peer_to_peer,3_peer_to_peer, 4_peer_to_peer}. For instance, in \cite{1_peer_to_peer}, by combining a voltage sensitivity analysis method with a distributed ledger technology, a decentralized peer-to-peer energy trading scheme is proposed to guarantee the energy transactions between users do not violate voltage constraints in a radial distribution network. In contrast to prior work, this paper studies a decentralized energy trading framework between a CES system and users with PV power generation by enabling the users and the CES provider to selfishly maximize their economic benefits while maintaining the network voltages within permissible limits.
 
\section{System Models of the Energy Trading System}\label{Sys_models}

The energy trading system comprises energy users and a CES system as depicted in Fig.~\ref{fig:config_system}. The CES system is owned by a third party, referred to as the CES provider, that provides storage services \cite{CES_ex}. This section first explains the demand-side models, describing the roles of the energy users, followed by the CES model, and the energy cost models. Then the distribution network power flow model is presented. Before integrating the demand-side and the CES models with the network power flow model, a generic representation of energy transactions among the key system entities, as shown in Fig.~\ref{fig:config_system}, is used for the clarity of explaining the role of each entity. 
\begin{figure}[b!]
\centering
\includegraphics[width=0.75\columnwidth]{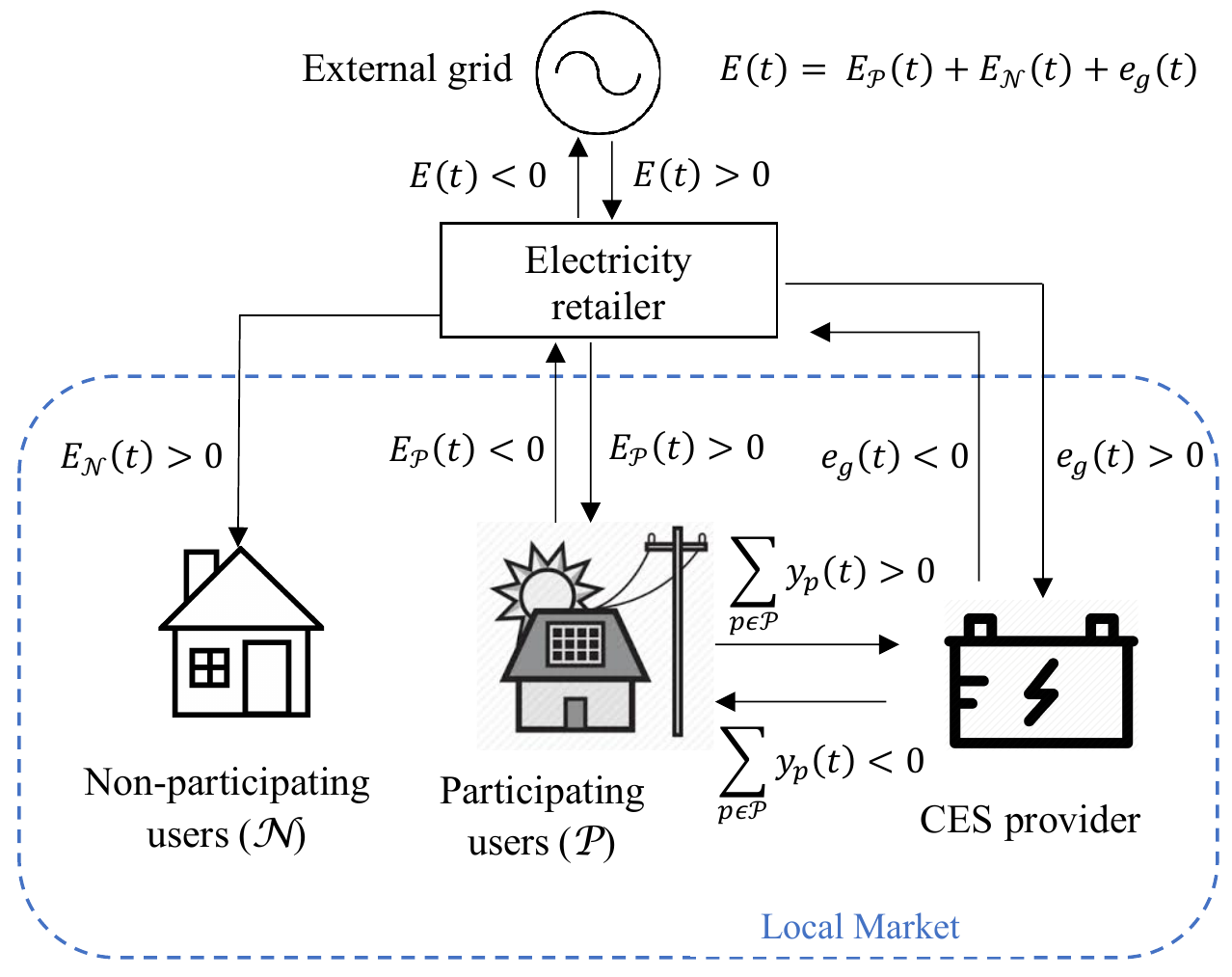}
\caption{A representation of the energy transactions among the system entities.}
\label{fig:config_system}
\end{figure}

\subsection{Demand-side Model}\label{Demand_model}
The demand-side of the energy trading system comprises two sets of users; participating users $\mathcal{P}$ and non-participating users $\mathcal{N}$, and $\mathcal{A} = \mathcal{P} ~\cup~\mathcal{N}$. The households of the users $\mathcal{P}$ are equipped with rooftop PV systems, operated at unity power factor, without behind-the-meter ESSs. The users $\mathcal{P}$ consume energy generated from their rooftop PV systems first to supply energy demand. Then if there is an energy deficit, the users $\mathcal{P}$ may decide to buy energy either from the external grid and/or from the CES system. If there is surplus PV energy, the users $\mathcal{P}$ may decide to sell some or all that energy to the CES system in addition to selling to the external grid. Likewise, the users $\mathcal{P}$ participate in the energy trading optimization framework, and it is considered that the households of the users $\mathcal{P}$ are equipped with controlling devices that make the energy trading decisions on their behalf. The users $\mathcal{N}$ do not participate in the energy trading optimization framework and are considered to be the traditional grid energy users without any local power generation or energy storage. The retailer, as shown in Fig.~\ref{fig:config_system}, acts as the middle-man between the external grid and the local market and coordinates the external grid energy transactions of the users $\mathcal{A}$ and the CES provider for billing purposes. 
 
The time period, typically one day, is divided into $H$ time intervals with length $\Delta t$, and $\mathcal{T} = \{1,2,\dotsm,H\}$. By considering the PV energy generation and energy demand fluctuations, the users $\mathcal{P}$ are sub-divided into two time-dependent sets; $\mathcal{P}_{\scaleto{+}{4.5pt}}(t)$ and $\mathcal{P}_{\scaleto{-}{4.5pt}}(t)$ such that $\mathcal{P} =\mathcal{P}_{\scaleto{+}{4.5pt}}(t)\cup \mathcal{P}_{\scaleto{-}{4.5pt}}(t)$. We take $d_a(t)\geq 0$ for each user $a \in \mathcal{A}$, and $g_p(t)\geq 0$ for each user $p\in \mathcal{P}$.
Surplus energy at user $p \in \mathcal{P}$ is given by
\begin{equation}
s_p(t) = g_p(t) - d_p(t). \label{eq:id1a}
\end{equation}
If $s_p(t)~\geq~0$, then user $p \in \mathcal{P}_{\scaleto{+}{4.5pt}}(t)$, and if $s_p(t)<0$, then user $p \in \mathcal{P}_{\scaleto{-}{4.5pt}}(t)$. Each user $p\in \mathcal{P}$ evaluates $y_p(t)$ based on $s_p(t)$. These strategies are evaluated day-ahead by using the day-ahead PV energy generation and demand forecasts, and we assume that the users $\mathcal{P}$ have accurate forecasts. $y_p(t) > 0$ if the user sells energy to the CES system, and $y_p(t) < 0$ if they buy energy from the CES system. 
It is specified that
\begin{subequations}\label{eq:id2a}
\begin{gather}
               0 \leq  y_p(t) \leq  s_p(t),~~~\text{if}~p \in \mathcal{P}_{\scaleto{+}{4.5pt}}(t), \label{subeq2a_1}\\
               s_p(t) \leq y_p(t) \leq 0,~~~\text{if}~p \in \mathcal{P}_{\scaleto{-}{4.5pt}}(t). \label{subeq2a_2}            
\end{gather}
\end{subequations}
Additionally, $e_p(t) > 0$ if user $p \in \mathcal{P}$ buys energy from the external grid, and $e_p(t) < 0$ if they sell energy to the external grid. Then, from the energy balance at user $p \in \mathcal{P}$, 
\begin{equation}
s_p(t) = y_p(t) - e_p(t). \label{eq:id2aa}
\end{equation}

\subsection{Community Energy Storage Model}\label{SES_model}
The CES provider computes $e_g(t)$ where $e_g(t) > 0$ if they buy energy from the external grid, and $e_g(t) < 0$ if they sell energy to the external grid. Given the energy trading decisions $y_p(t)$ and $e_g(t)$ made by the users $\mathcal{P}$ and the CES provider, respectively, actual energy flowed into/out of the CES system at time $t$ is calculated by 
\begin{equation}
e_s(t) =  e_g(t) + \sum_{p\in \mathcal{P}}y_p(t). \label{eq:id1}
\end{equation}
In \eqref{eq:id1}, if $e_s(t) > 0$, then the CES system is charging, and if $e_s(t) < 0$, then the CES system is discharging. To satisfy the storage charging/discharging power rates, it is required that
\begin{equation}
-\gamma^{\text{dis}}_{\text{max}} \leq \frac{e_s(t)}{\Delta t }\leq \gamma^{\text{ch}}_{\text{max}}, ~\forall{t \in \mathcal{T}}. \label{eq:id2}
\end{equation}

By taking $0<\eta_c \leq 1$  and $\eta_d \geq 1$, and by denoting the energy charge level of the CES system at the beginning of time $t$ as $b(t-1)$, $b(t)$ is given by
\begin{equation}
b(t) = \begin{cases}
b(t-1) + \eta_c e_s(t), &\text{if $e_s(t) \geq 0$},\\
b(t-1) + \eta_d e_s(t), &\text{if $e_s(t) <0$.}
\end{cases}
\end{equation} \label{eq:id3}
To ensure that the energy charge levels of the CES system remain within the storage capacity limits, it is specified
\begin{equation}
B_{\text{min}}\leq b(t) \leq B_{\text{max}},~\forall{t \in \mathcal{T}}. \label{eq:id4}
\end{equation}

Finally, to ensure the continuous operation of the CES system for the next day and to avoid over-charging or over-discharging of the CES system, it is considered that the initial charge level $b(0)$ would be approximately the same as the charge level $b(H)$ at the end of time $H$ \cite{Atzeni}. Then
\begin{equation}
\abs{b(H) - b(0)} \leq  \theta. \label{eq:id7}
\end{equation}

\subsection{Energy Pricing and Cost Models} \label{cost_model}
 In the system, the external grid energy price imposed by the retailer at time $t$ is given by
\begin{equation}
\lambda_g(t) = \phi_t E(t) + \delta_t \label{eq:id8}
\end{equation}
where $\phi_t,~\delta_t > 0$ can be obtained through a day-ahead electricity market clearing process \cite{Atzeni}. The dynamic price function \eqref{eq:id8} is widely used in the smart grid literature and can be used to encourage users to shift their peak energy demand to non-peak hours \cite{Lambotharan,Gupta,Atzeni}. $E(t) = E_{\mathcal{P}}(t) + E_{{\mathcal{N}}}(t) + e_g(t)$ given that $E_{\mathcal{P}}(t) = \sum_{p\in \mathcal{P}}e_p(t)$ and $E_{\mathcal{N}}(t) = \sum_{n\in \mathcal{N}}d_n(t)$. At time $t$, $E(t)$ can be either positive or negative. For instance, if $E_{\mathcal{P}}(t) < 0$, due to all users in $\mathcal{P}$ having positive surplus energy, i.e., $s_p(t) > 0,~\forall p\in \mathcal{P}$, and if  $|E_{\mathcal{P}}(t)| > |E_{{\mathcal{N}}}(t) + e_g(t)|$, then $E(t) < 0$. With the existence of negative $E(t)$, the price calculation in \eqref{eq:id8} can be negative. Hence, to ensure positive grid prices at each time $ t $, it is specified 
\begin{equation}
 \lambda_{\text{min}}  \leq \lambda_g(t). \label{eq:id8aa}
\end{equation}
Similar to \cite{Gupta}, it is supposed that $-E_{\text{rev,max}} \leq E(t) \leq E_{\text{fw,max}}$ without overloading the distribution transformer.

The retailer adopts a one-for-one non-dispatchable energy buyback scheme \cite{Solar_scheme} where it applies the same price for both selling and buying external grid energy transactions of the users $\mathcal{P}$ and the CES provider. For instance, if user $p \in \mathcal{P}$ buys $e_p(t)$ energy from the external grid, they incur an energy cost of $\lambda_g(t)e_p(t)$ whereas if user $p$ sells $e_p(t)$ energy to the grid, then they receive $\lambda_g(t)e_p(t)$ payment from the retailer.

In the system, the CES provider sets a time-dependent price $\lambda_s(t) >0 $ for the energy transactions between the users $\mathcal{P}$ and the CES system. Then $C_p (t)$ for user $p\in \mathcal{P}$ is given by
\begin{equation}
C_p (t)= \lambda_g(t)e_p(t) - \lambda_s(t)y_p(t). \label{eq:id9}
\end{equation}

Since the CES provider trades energy $e_g(t)$ with the retailer at price $\lambda_g(t)$ and energy $y_p(t)$ with the users $\mathcal{P}$ at price $\lambda_s(t)$, their total revenue is given by
\begin{equation}
W_s= \sum_{t=1}^H\Big( - \lambda_s(t)\sum_{p\in \mathcal{P}}y_p(t) - \lambda_g(t)e_g(t)\Big). \label{eq:id10}
\end{equation}

\subsection{Distribution Network Power Flow Model}\label{PF_theory}
We consider integrating the demand-side and CES models in Sections~\ref{Demand_model} and \ref{SES_model} with a radial distribution network and employ the Distflow equations developed in \cite{Baran1} to model the network power flow. For simplicity, time index $t$ is omitted in power and voltages in equations \eqref{eq:id16_a} - \eqref{eq:id22_a}.  

Consider a radial distribution network described by $\mathcal{G}=(\mathcal{V}, \mathcal{E})$ with $\mathcal{V}= \{0,1,\dotsm,\text{N}\}$ and $\mathcal{E} = \{(i,j)\} \subset \mathcal{V}\times \mathcal{V}$. Root bus 0 represents the secondary of the distribution transformer and is considered to be the slack bus. The natural radial network orientation is considered where every distribution line points away from bus 0 \cite{Bitar}. It is supposed that $V_0$ is fixed and known. The Distflow equations are given by, for each bus $ j = \{1,\dotsm,\text{N}\}$ and their parent bus $i \in \mathcal{V}$,
\begin{subequations}\label{eq:id16_a}
\begin{gather}
 P_{ij} = P_{j} + \sum_{k: (j,k) \in \mathcal{E}} P_{jk} + r_{ij}\ell_{ij},\\
Q_{ij} = Q_{j} + \sum_{k: (j,k) \in \mathcal{E}} Q_{jk} + x_{ij}\ell_{ij},\\
V^2_{j} = V^2_{i} - 2(r_{ij}P_{ij}+x_{ij}Q_{ij})+ (r^2_{ij}+x^2_{ij})\ell_{ij}
\end{gather}
\end{subequations}
where $\ell_{ij} = \frac{P_{ij}^2 + Q_{ij}^2 }{V^2_{i}}$. We exploit the linear approximation of the power flow model \eqref{eq:id16_a}  developed in \cite{Baran1}. The linear approximation relies on the assumption that $\ell_{ij}=0$, and has been extensively justified and used in power flow calculations in radial distribution networks \cite{MG1,Bitar}. The compact form of the linearized Distflow equations is given by \cite{Bitar}
\begin{equation}
 \mathbf{V} = \mathbf{R}\mathbf{P} + \mathbf{X}\mathbf{Q} + V_0^2\mathbf{1} \label{eq:id22_a}
\end{equation} 
where $ \mathbf{V} = (V_1^2,\dotsm,V_{\text{N}}^2)^{\text{T}}$, ~$\mathbf{P} = (P_1,\dotsm,P_{\text{N}})^{\text{T}}$ and $\mathbf{Q} = (Q_1,\dotsm,Q_{\text{N}})^{\text{T}}$. Additionally, $\mathbf{1}$ is the N-dimensional vector of all ones. $\mathbf{R}, ~\mathbf{X}\in \mathbb{R}^{\text{N}\times \text{N}}$ where $\mathbf{R}_{ij} = -2\sum_{(h,k)\in \mathcal{J}_i \cap \mathcal{J}_j}r_{hk}$, ~$\mathbf{X}_{ij} = -2\sum_{(h,k)\in \mathcal{J}_i \cap \mathcal{J}_j}x_{hk}$, and $\mathcal{J}_i \subset \mathcal{E}$. It is required that the vector of voltage magnitudes $\mathbf{V}$ in \eqref{eq:id22_a} to satisfy
\begin{equation}
\mathbf{V}_{\text{min}}~\leq~\mathbf{V}~\leq~\mathbf{V}_{\text{max}},~\forall t\in \mathcal{T}. \label{eq:id23_a}
\end{equation}
In \eqref{eq:id23_a}, $\mathbf{V}_{\text{min}} \in \mathbb{R}^{\text{N}\times 1}$ where all entries being $V_{\text{min}}^2$ and $\mathbf{V}_{\text{max}} \in \mathbb{R}^{\text{N}\times 1}$ where all entries being $V_{\text{max}}^2$.

We consider, at each bus $i\in \mathcal{V}\backslash \{0\}$, $P_i$ comprises the total active power consumption and $Q_i$ comprises the total reactive power consumption of both participating and non-participating users at bus $i$, at time $t$. Additionally, if the CES system is at bus $i\in \mathcal{V}\backslash \{0\}$, the CES active power consumption, $\frac{e_s(t)}{\Delta t}$, is also included in $P_i$. It is taken that the PV and CES inverters operate at unity power factor. Hence, if the CES system is at bus $i\in \mathcal{V}\backslash \{0\}$, then $P_{i} = \frac{1}{\Delta t} \Big( -\sum_{p\in \mathcal{P}_i}s_p(t) + \sum_{n\in \mathcal{N}_i}d_n(t)  + e_s(t)\Big)$. Otherwise, $P_{i} = \frac{1}{\Delta t} \Big( -\sum_{p\in \mathcal{P}_i}s_p(t) + \sum_{n\in \mathcal{N}_i}d_n(t) \Big)$. Here, $s_p(t)$ can be calculated by using \eqref{eq:id1a}. Additionally, $\mathcal{P}_i \subset \mathcal{P}$ and $\mathcal{N}_i \subset \mathcal{N}$. The reactive power consumption $Q_i$ in \eqref{eq:id22_a} can be calculated by using $Q_i= \sum_{a\in \mathcal{A}_i}q_a(t)$ where $ \mathcal{A}_i = \mathcal{P}_i \cup \mathcal{N}_i $. 

\begin{remark}
Power flows may be restricted by line flow limits given by $P^2_{ij} +   Q_{ij}^2\leq S^2_{ij,\text{max}},~\forall (i,j) \in \mathcal{E}$ \cite{Taylor}. 
These line flow constraints can be included as quadratic constraints in the optimization problem in Section~\ref{Obj_SES} which can then be solved as a quadratically-constrained quadratic program \cite{boyd2004convex}.
\end{remark}

\section{Game-theoretic Formulation of the Energy Trading System} \label{sec_system}
The energy trading interactions between the CES provider and the users $\mathcal{P}$ are studied using a non-cooperative Stackelberg game. Generally, in a Stackelberg game, one player acts as the leader and moves first to make decisions, and the rest of the players follow the leader's decisions to make their own decisions. Here, the CES provider acts as the leader and determines the optimal values for $(\lambda_s(t),~e_g(t))$. The users $\mathcal{P}$ are the followers and determine the optimal values for $y_p(t)$. We exploit the backward induction method \cite{gametheoryessentials} where the followers' actions, i.e., $y_p(t)$, are determined first using the knowledge of leader's actions, i.e., $(\lambda_s(t),~e_g(t))$, and then the analysis proceeds backwards to determine the leader's actions.\footnote{To make the game-theoretic analysis tractable, in this paper, we assume accurate forecasts of PV power generation and electricity demand, and the perfect knowledge of player actions. To handle imperfect information from inaccurate energy forecasts and from not having access to perfect knowledge of actions of other players, an interesting future work can incorporate stochastic game theory with imperfect information with the game-theoretic energy trading framework  \cite{gametheoryessentials}.} This process is explained in the next two subsections. 

\subsection{Objectives of the Participating Users}\label{Obj_par}
The primary objective of user $p \in \mathcal{P}$ is to minimize their individual energy cost \eqref{eq:id9} at each time $t \in \mathcal{T}$ given feasible values for $(\lambda_s(t),~e_g(t))$. By substituting \eqref{eq:id8} in \eqref{eq:id9} with \eqref{eq:id2aa}, \eqref{eq:id9} can be written as a quadratic function of $y_p(t)$ as 
\begin{equation}
C_p(y_p(t)) = K_2y_p(t)^2 + K_1y_p(t) + K_0.  \label{eq:id11}
\end{equation} 
Here, $K_2 = \phi_t,~K_1= -\big(\phi_t(2s_p(t)-E_{-p}(t))-\delta_t + \lambda_s(t)\big)$, $K_0 = -\lambda_s(t)s_p(t)$ with $E_{-p}(t) = E_{{\mathcal{P}}\backslash p} (t) + E_{\mathcal{N}}(t) +  e_g(t)$. $E_{{\mathcal{P}}\backslash p} (t)$ denotes total external grid energy of the other users $\mathcal{P}\backslash p$, i.e., $E_{{\mathcal{P}}\backslash p} (t) = E_{\mathcal{P}}(t)-e_p(t)$. 

Since the cost function \eqref{eq:id11} depends on the actions of the other users $\mathcal{P}\backslash p$, a non-cooperative game that can be played by the users $\mathcal{P}$ at time $t$ is formulated to determine the optimal values for $y_p(t)$. The strategic form of the game among the users $\mathcal{P}$ at time $t$ is denoted by  $\Gamma \equiv \langle \mathcal{P},~\mathcal{Y},~\mathcal{C}\rangle $ where $\mathcal{Y} = \prod_{p \in \mathcal{P}} \mathcal{Y}_p (t)$ and $\mathcal{C} = \{C_1(t),\dotsm,C_M(t)\}$. Here, $ \mathcal{Y}_p (t)$ is subject to constraints \eqref{eq:id2a}. Suppose the CES energy transactions profile of the users $\mathcal{P}$ at time $t$ as $\ve{y}(t) = (y_1(t),\dotsm,y_M(t))$, and the CES energy transactions profile of the other users $\mathcal{P}\backslash p$ as $\ve{y}_{\mathcal{P}\backslash p}(t) = (y_1(t),\dotsm,y_{p-1}(t),y_{p+1}(t),\dotsm,y_M (t))$. Then $C_p(t) \equiv C_p(y_p(t),\ve{y}_{\mathcal{P}\backslash p}(t))$.
At time $t$, user $p\in \mathcal{P}$ determines
\begin{equation}
\tilde y_p(t) = \argmin_{y_p(t)~\in~\mathcal{Y}_p(t)}~C_n(y_p(t),~\ve{y}_{\mathcal{P}\backslash p}(t)). \label{eq:id12}
\end{equation} 

\begin{proposition}\label{prop1a}
The game $\Gamma$ has a unique pure strategy Nash equilibrium for any given feasible values of $(\lambda_s(t),~e_g(t))$.
\end{proposition}

\begin{proof}
At a Nash equilibrium, no user can benefit by unilaterally changing their own strategy while the other users play their Nash equilibrium strategies \cite{gametheoryessentials}. Clearly, because $\phi_t >0$, the second derivative of \eqref{eq:id11} with respect to $y_p(t)$ is positive and therefore, \eqref{eq:id11} is strictly convex for given feasible strategy profile of $\ve{y}_{\mathcal{P}\backslash p}(t)$. Therefore, the objective function in \eqref{eq:id12} is strictly convex. Also, the individual strategy sets $\mathcal{Y}_p(t)$ are compact and convex as they are subject to linear inequalities \eqref{eq:id2a} \cite{boyd2004convex}. Thus, the existence of a unique Nash equilibrium with pure strategies for the game $\Gamma$ is guaranteed \cite{rosen}.
\end{proof}
In the game $\Gamma$, the best response of each user $p\in \mathcal{P}$ for given $\ve{y}_{\mathcal{P}\backslash p}(t)$ is found by using
\begin{equation}
 \frac{\partial{C_p(y_p(t))}}{\partial{y_p(t)}} = 2K_2\tilde y_p(t) + K_1 = 0. \label{eq:id13}
\end{equation} 
Then to find the Nash equilibrium solutions, \eqref{eq:id13} is solved for all users in $\mathcal{P}$ and that leads to solving $M$ number of simultaneous equations. Once these $M$ equations are solved with the expressions $K_2$ and $K_1$ in \eqref{eq:id11}, the optimal response of user $p \in \mathcal{P}$ at the Nash equilibrium $y^*_p(t)$ is given by
 \begin{equation}
y^*_p(t)= s_p(t) ~+~\epsilon(t) \label{eq:id14}
\end{equation}
where $\epsilon(t) = (M+1)^{-1}\big[\phi^{-1}_t(\lambda_s(t)-\delta_t) - E_{\mathcal{N}}(t) - e_g(t)\big]$.  However, to form the Nash equilibrium of the game $\Gamma$, $y^*_p(t)$ should also satisfy constraints \eqref{eq:id2a}. It can be seen that $ y^*_p(t)$ in \eqref{eq:id14} are functions of the CES provider's strategies $(\lambda_s(t),~e_g(t))$. Thus, any values for $(\lambda_s(t),~e_g(t))$ would not guarantee that $y^*_p(t)$ satisfies constraints \eqref{eq:id2a}. Hence, to ensure $y^*_p(t)$ satisfies \eqref{eq:id2a} for each user $p \in \mathcal{P}$, auxiliary constraints \eqref{eq:id16} are considered in the revenue maximization problem of the CES provider\footnote{The consideration of the auxiliary constraints \eqref{eq:id16} in the CES provider's revenue optimization problem is a potential implementation where, depending on the nature of the users $\mathcal{P}$, the CES provider only needs to know the maximum or the minimum surplus energy amount of the users $\mathcal{P}$. In an alternative method, in response to $(\lambda_s(t),~e_g(t))$, the users $\mathcal{P}$ may solve their individual optimization problems \eqref{eq:id12} iteratively by incorporating the constraints \eqref{eq:id2a} until the Nash equilibrium is reached which requires additional computation time for the iterative negotiation.} as described in the next subsection.

\subsection{Objective of the Community Energy Storage Provider} \label{Obj_SES}
As per backward induction, the aggregated Nash equilibrium CES energy amounts of the users $\mathcal{P}$, i.e., $\sum_{p\in \mathcal{P}} y^*_p(t)$, can be substituted in the revenue function \eqref{eq:id10} which can then be written in terms of the CES provider's actions $(\lambda_s(t),~e_g(t))$. Then the CES provider's revenue maximization problem is to determine
\begin{subequations}\label{eq:id15}
\begin{gather}
\ve{\rho^*} = \argmax_{\ve{\rho}~\in~\mathcal{U}}\sum_{t=1}^H\big( \mu_1 \lambda_s(t)^2 + \mu_2 \lambda_s(t) +  \mu_3 e_g(t)^2 + \mu_4 e_g(t) \big),\\
\intertext{where constant coefficients, $ \mu_1,~ \mu_2,~ \mu_3,~ \mu_4$, are given by}
\mu_1 =  \frac{-M}{\phi_t~(M+1)}, \label{eq:id15_coeff1}\\
\mu_2 = \frac{M}{(M+1)} \Big(E_{\mathcal N}(t) + \frac{\delta_t}{\phi_t}\Big) -\sum_{p \in \mathcal{P}}s_p(t), \label{eq:id15_coeff2}\\
\mu_3 = \frac{-\phi_t}{(M+1)},\label{eq:id15_coeff3}\\
\mu_4 = \frac{-(\phi_tE_{\mathcal{N}}(t) + \delta_t)}{(M+1)}\label{eq:id15_coeff4}.
\end{gather}
\end{subequations}
Additionally, the matrix of decision variables of the CES provider $\ve{\rho} = (\ve{\lambda_s}, \ve{e_g})$ with $\ve{\lambda_s} = (\lambda_s(1),\dotsm,\lambda_s(H))^{\text{T}}$ and $\ve{e_g} = (e_g(1),\dotsm,e_g(H))^{\text{T}}$. 
As realized in Section~\ref{Obj_par}, to ensure that $y_p^*(t)$ satisfies \eqref{eq:id2a}, the CES provider selects $(\lambda_s(t),e_g(t))$ at time $t$ such that
\begin{equation}
\begin{rcases}
              - \text{min}[\ve{s}(t)]\leq \epsilon(t)\leq 0,~~~~~\text{if all}~\mathcal{P}~\text{are surplus}, \\
             0 \leq \epsilon(t) \leq  -\text{max}[\ve{s}(t)],~~~~~\text{if all}~\mathcal{P}~\text{are deficit},\\
              \epsilon(t)=0,~~~~~~~~~~~~\text{if}~\mathcal{P}~\text{has both types of users}.
\end{rcases}\label{eq:id16}
\end{equation}
where $\epsilon(t)$ is given in \eqref{eq:id14}. $\text{min}[\ve{s}(t)]$ takes the minimum value of the surplus energy profile of the users $\mathcal{P}$, i.e, $\ve{s}(t) = (s_1(t),\dotsm,s_M(t))$, and $\text{max}[\ve{s}(t)]$ takes its maximum value. The strategy set $\mathcal{U}$ is subject to constraints \eqref{eq:id2}, \eqref{eq:id4}, \eqref{eq:id7}, \eqref{eq:id8aa}, \eqref{eq:id23_a} and \eqref{eq:id16}. Since the Hessian is negative definite for all $ \ve{\rho} \in \mathcal{U}$ as $ \mu_1,~\mu_3<0$, the objective function in \eqref{eq:id15} is strictly concave. Additionally, $\mathcal{U}$ is non-empty, closed, and convex as it is subject to linear constraints. Hence, the optimization \eqref{eq:id15} has a unique solution \cite{boyd2004convex}.  

\begin{remark}
The reasons for choosing the lower and upper bounds of the constraints \eqref{eq:id16} are as follows; By comparing \eqref{eq:id14} with \eqref{eq:id2aa}, at the Nash equilibrium, $\epsilon(t) = e_p(t)$ and that gives us $ y^*_p(t)= s_p(t) + e^*_p(t)$ with $e^*_p(t)$ being the traded grid energy of user $p\in \mathcal{P}$ at the Nash equilibrium. If all users in $\mathcal{P}$ are surplus users at time $t$, i.e., $s_p(t) \geq 0,~\forall p \in \mathcal{P}$, then to ensure constraint \eqref{subeq2a_1} is satisfied, it is required that $ 0 \leq  y^*_p(t)= (s_p(t) + e^*_p(t)) =  (s_p(t) + \epsilon(t))\leq s_p(t), ~\forall p \in \mathcal{P}$. Hence, the choice of $\epsilon(t)$ above $- \text{min}[\ve{s}(t)]$ and below $0$ guarantees, \eqref{subeq2a_1} is satisfied for all users in $\mathcal{P}$. A similar argument is applied when there are only deficit users at time $t$, i.e, $s_p(t) < 0,~\forall p \in \mathcal{P}$, where constraint \eqref{subeq2a_2} has to be satisfied. If time $t$ has both types of users, $e^*_p(t)$ is set to zero, and therefore, $\epsilon(t)=0$, so that $y^*_p(t) = s_p(t), \forall p \in \mathcal{P}$. 
\end{remark}

\subsection{Non-cooperative Stackelberg Game}\label{Stack_game}
The strategic form of the non-cooperative Stackelberg game between the CES provider and the users $\mathcal{P}$ is given as $\Theta  \equiv \langle \{\mathcal{L}, \mathcal{F}\}, \{\mathcal{U}, \mathcal{Y}\}, \{W_{s}, \mathcal{C} \}\rangle $. Here, $\mathcal{L}$ is the CES provider, and $\mathcal{F}$ are the users $\mathcal{P}$. A suitable solution for the proposed game $\Theta$ is the Stackelberg equilibrium in which the leader attains their optimal price and grid energy given the followers' equilibrium state. In game theory context, a Stackelberg equilibrium is a stable solution at which none of the players, i.e., the leader or any follower, can benefit by altering their strategy unilaterally. Additionally, in non-cooperative games, it is not always guaranteed to exist an equilibrium in pure strategies \cite{basardynamic}. Proposition \ref{prop2a} below along with Proposition \ref{prop1a} in Section~\ref{Obj_par} guarantees that there exists a \textit{unique Stackelberg equilibrium in pure strategies} for the energy trading game $\Theta$.
 
\begin{proposition}\label{prop2a}
The game $\Theta$ has a unique pure strategy Stackelberg equilibrium.
\end{proposition}

\begin{proof}
For given feasible $(\lambda_s(t), e_g(t))$, the non-cooperative game $\Gamma$ has a unique Nash equilibrium for the energy transactions $y_p(t)$ of the users $\mathcal{P}$ and that is given by $y_p^*(t)$  (see Proposition~\ref{prop1a}). By incorporating $\sum_{p\in \mathcal{P}} y^*_p(t)$, the CES provider's revenue maximization also has a unique solution as proved in Section \ref{Obj_SES}. Thus, as per backward induction, the game $\Theta$ has a unique Stackelberg equilibrium $(\ve{\rho^*}, \ve{y}^*)$ where $\ve{y}^* = (\ve{y}^*(1),\dotsm,\ve{y}^*(H))$, and $\ve{y}^*(t)$ is found by substituting the $\text{t}^{th}$ element in $\ve{\rho^*}$, i.e., $\ve{\rho^*}(t) = (\lambda_s^*(t), e_g^*(t))$, in \eqref{eq:id14} for each user $n\in \mathcal{P}$ \cite{gametheoryessentials}. 
\end{proof}

Note that the equilibrium $(\ve{\rho^*}, \ve{y}^*)$ of the game $\Theta$ satisfies
\begin{multline}
C_p( \ve{y}^*(t),\ve{\rho^*}) \leq C_p((y_p(t), \ve{y}^*_{{\mathcal{P}\backslash p}}(t)), \ve{\rho^*}),\\
~\forall p\in \mathcal{P}, \forall y_p(t) \in \mathcal{Y}_p(t),~\forall t\in \mathcal{T} ,\label{eq:id17a}
\end{multline}
\begin{equation}
 \begin{split}
W_s(\ve{y}^*,\ve{\rho^*}) \geq W_s(\ve{y}^*,\ve{\rho}),~\forall \ve{\rho} \in \mathcal{U}.
\end{split}\label{eq:id18a}
\end{equation}
In \eqref{eq:id17a}, $\ve{y}^*_{{\mathcal{P}\backslash p}}(t)$ is the Nash equilibrium strategy profile of the other users $\mathcal{P}\backslash p$ at time $t$, i.e., $\ve{y}^*_{{\mathcal{P}\backslash p}}(t) = (y^*_1(t),\dotsm,y^*_{p-1}(t),y^*_{p+1}(t),\dotsm,y^*_M (t))$. 

To implement the game $\Theta$, the CES provider solves \eqref{eq:id15} and then, announces $\ve{\rho^*}$ to the users $\mathcal{P}$ to find their Nash equilibrium solutions using \eqref{eq:id14}. To solve \eqref{eq:id15}, the CES provider needs to know the aggregated surplus energy of the users $\mathcal{P}$, i.e., $\sum_{p \in \mathcal{P}}s_p(t)$ to calculate the objective function coefficient \eqref{eq:id15_coeff2} and to calculate $e_s(t)$ for the constraints \eqref{eq:id2}, \eqref{eq:id4}, \eqref{eq:id7}, and \eqref{eq:id8aa} in $\mathcal{U}$. Note that, from \eqref{eq:id2aa} and \eqref{eq:id1}, $e_s(t)$ is a function of $\sum_{p \in \mathcal{P}}s_p(t)$. Additionally, the CES provider requires the information of the maximum or the minimum surplus energy amount of the users $\mathcal{P}$ as required by \eqref{eq:id16}, and the aggregated surplus energy and the aggregated reactive power demand of the users $\mathcal{P}_i,~\forall i \in \mathcal{V}$, i.e., $\sum_{p\in \mathcal{P}_i}s_p(t)$ and $\sum_{p\in \mathcal{P}_i}q_p(t)$, as required by the voltage constraint \eqref{eq:id23_a}. As such, the disclosure of individual strategies or energy usage information of the users $\mathcal{P}$ to the CES provider is not required.

\section{Results and discussion}\label{sec:5}
We consider the realistic 7-bus radial feeder in \cite{feeder_model} with a 22/0.4 kV, 185 kVA distribution transformer (see Fig.~\ref{fig:LV_feeder}). 
\begin{figure}[b!]
\centering
\includegraphics[width=0.85\columnwidth]{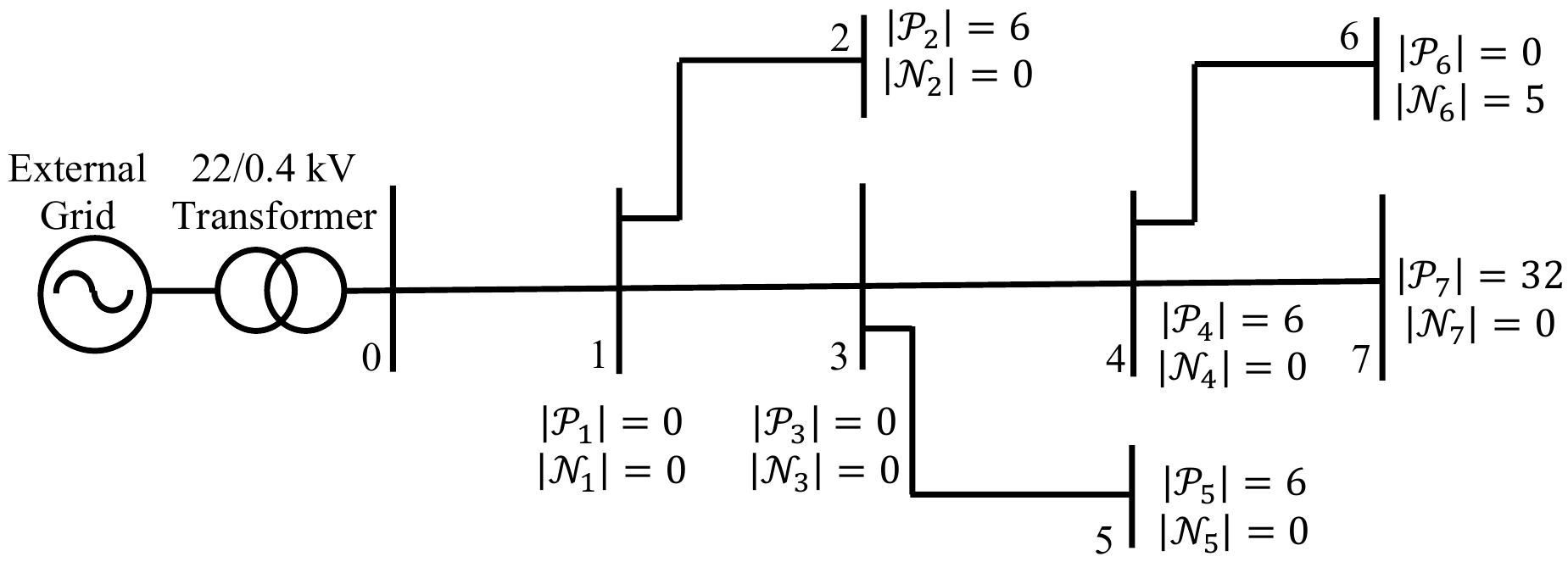}
\caption{One line diagram of the radial feeder. ~$|\mathcal{P}_{i}|$ - Number of participating users at bus $i$,~$|\mathcal{N}_{i}|$- Number of non-participating users at bus $i$.}
\label{fig:LV_feeder}
\end{figure}
It is considered that the feeder supplies 55 users and hence, $|\mathcal{A}| = 55$. As in \cite{feeder_model}, the secondary voltage of the transformer is set at 1.0 p.u., and the maximum and minimum voltage limits of the feeder are taken as 1.05 p.u. and 0.95 p.u., respectively. 
 
The active power demand and PV power generation profiles are chosen from a real dataset that includes 5-min PV power and demand measurements of a set of residential users in Canberra, Australia in 2018 \cite{NextGen}. For simulations in Sections~\ref{results_1},~\ref{results_3}, and \ref{centralized}, the daily PV power generation and demand profiles of the 55 users are generated such that they represent the average daily PV power and demand profiles of the selected 55 users in Autumn 2018. For each user, 92 daily PV and demand profiles were used to generate their average daily PV and demand profiles in Autumn.
Due to the lack of realistic reactive power demand data, reactive power demand of the users $\mathcal{A}$ is not considered. The total number of users at each bus of the feeder is calculated in proportion to the PV system allocation in \cite{feeder_model}. For instance, to represent that the bus 6 does not have PV systems, randomly selected 5 users from 55 users are allocated to bus 6 assuming that they are the users $\mathcal{N}$. The rest of the 50 users are considered as the users $\mathcal{P}$ and are allocated to other buses as shown in Fig.~\ref{fig:LV_feeder}. Additionally, $\Delta t = 5/60~\text{hrs}, ~H = 288$, $B_{\text{max}}= 700~ \text{kWh},~\gamma^{\text{ch}}_{\text{max}}=\gamma^{\text{dis}}_{\text{max}}=150~ \text{kW},~ \eta_c = 0.98,~\eta_d = 1.02,~B_{\text{min}}=0.05B_{\text{max}}$. $\delta_t$ is set as a constant equal to the average price of a reference two-step time-of-use (TOU) price signal in \cite{Origin} where the peak period is 07.00-23.00 (time intervals 85-276). $\phi_t$ is selected such that $\frac{\phi_{\text{peak}}}{\phi_{\text{off-peak}}} = \frac{\text{Price}_{\text{peak,ref-TOU}}}{\text{Price}_{\text{off-peak, ref-TOU}}} = 2.12$. Then $\phi_{\text{off-peak}}$ is set such that the difference between the peak and the off-peak prices of the reference TOU signal is equal to the difference between the predicted maximum grid price in the peak hours and the minimum grid price in the non-peak hours of the system. $\lambda_{\text{min}}$ is set at the reference TOU off-peak price $18.5 \text{~AU cents/kWh}$.

As realized in Section~\ref{Obj_SES} with backward induction, the CES provider's theoretical optimal revenue at the Stackelberg equilibrium can be obtained by optimizing \eqref{eq:id15}. Hence, in our simulations, to obtain the Stackelberg equilibrium of the game $\Theta$, first, \eqref{eq:id15} was solved for $\ve{\rho^*}$ using the interior point algorithm in MATLAB fmincon solver, and then the elements $\ve{\rho^*}(t) = (\lambda_s^*(t), e_g^*(t))$ were substituted in \eqref{eq:id14} to obtain $\ve{y}^*$. All simulations were conducted using MATLAB and a laptop with 2.7 GHz Intel Core i7 processor and 16 GB RAM. Under these settings, the computation time for finding $\ve{\rho^*}$ in simulations in Section~\ref{results_1} was 66.8 seconds. 

\subsection{Impact on Voltage Profiles and Energy Costs} \label{results_1}
Here, the impacts of the energy trading system on the bus voltages and the economic benefits for the CES provider and the users are compared with a baseline without a CES system. In the baseline, the users $\mathcal{P}$ trade energy only with the external grid through the retailer at price $\lambda_g(t)$ in \eqref{eq:id8}. In the energy trading system, the CES system is placed at bus 7 in Fig.~\ref{fig:LV_feeder}. 

Before the time 105 and after the time 209, all users in $\mathcal{P}$ are deficit users with little or zero PV energy generation, and that results in positive aggregate energy consumptions at each bus with the baseline as shown in Fig.~\ref{fig:enrgInj_buses_baseline}(a).
\begin{figure}[b!]
\centering
\includegraphics[width=0.75\columnwidth]{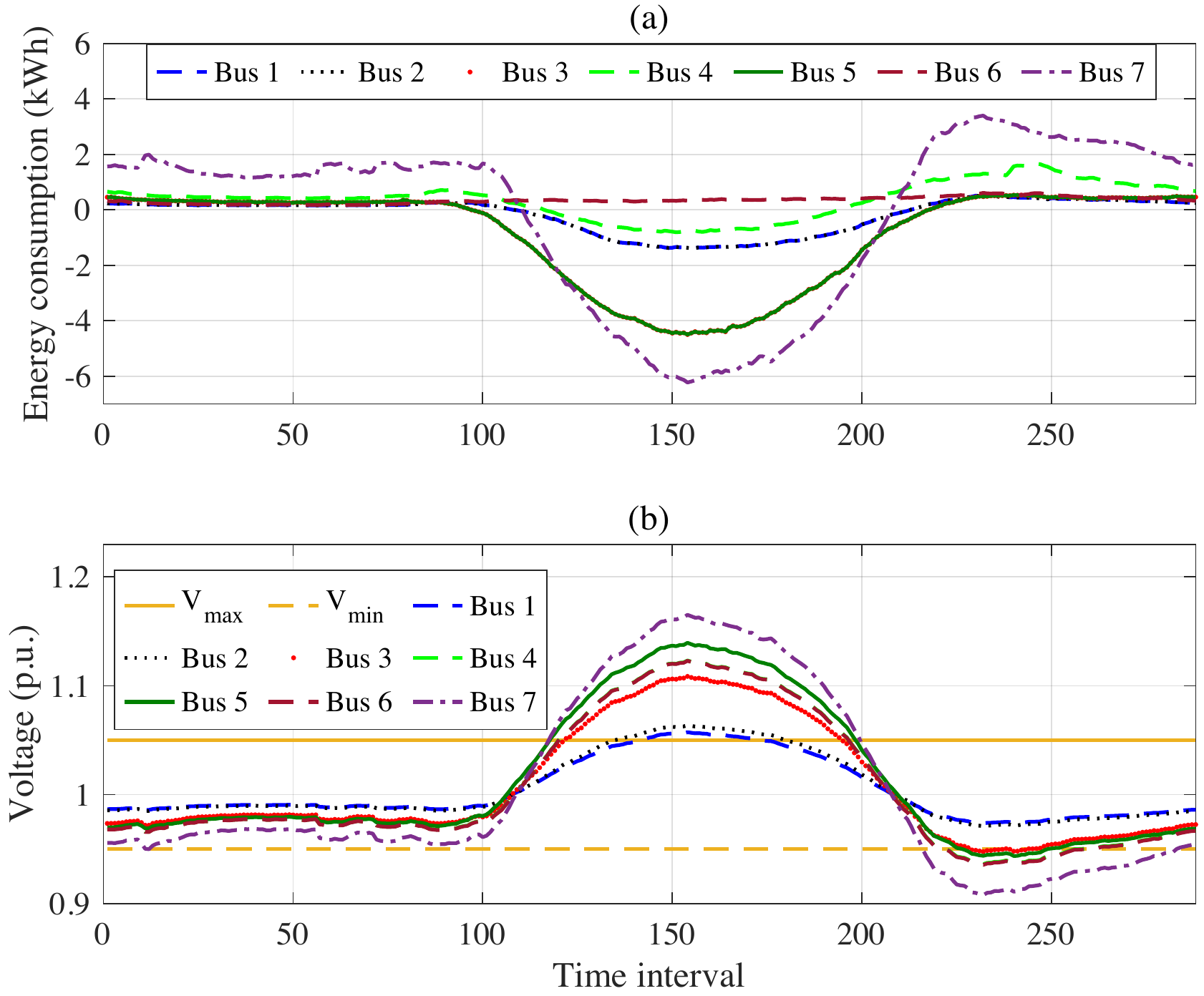}
\caption{(a) Aggregate bus energy consumptions and (b) bus voltages with the baseline.}
\label{fig:enrgInj_buses_baseline}
\end{figure}
Fig.~\ref{fig:Agg_GridLoad}(a) depicts, in the baseline, the aggregate positive electricity load on the external grid is greatest between times 216 and 281. 
\begin{figure}[t!]
\centering
\includegraphics[width=0.75\columnwidth]{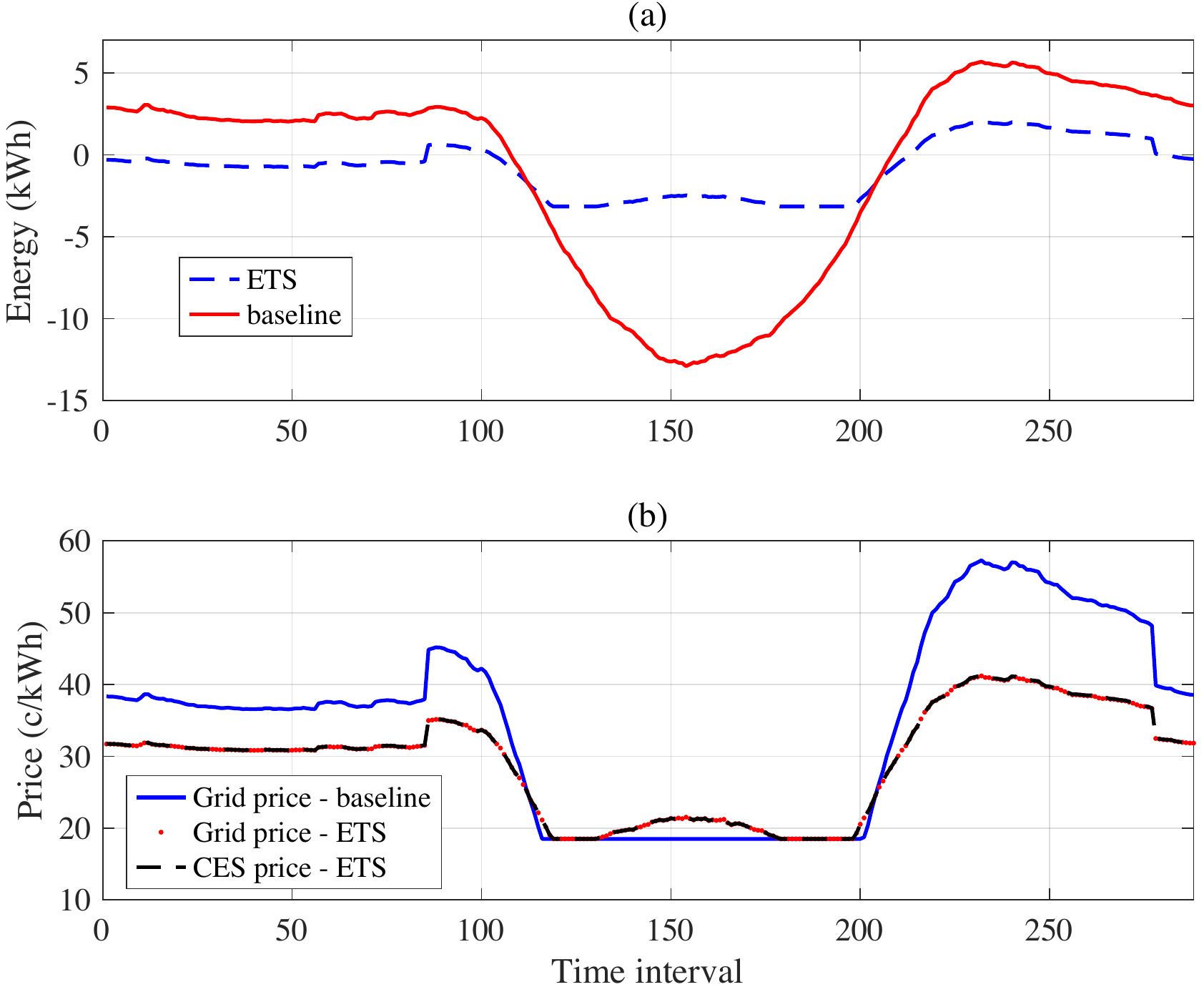}
\caption{(a)~Total external grid load $E(t)$ and (b) prices of the energy trading system (ETS) and the baseline.}
\label{fig:Agg_GridLoad}
\end{figure} 
Simultaneously, the voltages at buses 3 - 7 drop below the lower limit 0.95 p.u. causing under-voltages (see Fig.~\ref{fig:enrgInj_buses_baseline}(b)), and the lowest voltage $0.908$ p.u. occurs at bus 7.  As per Fig.~\ref{fig:enegy_transac.}, in the energy trading system, the users $\mathcal{P}$ buy more energy from the CES than from the external grid  $(|\sum_{p\in \mathcal{P}}y_p(t)| > |E_{\mathcal{P}}(t)|)$ to supply their peak energy demand after the time 209. 
\begin{figure}[t!]
\centering
\includegraphics[width=0.75\columnwidth]{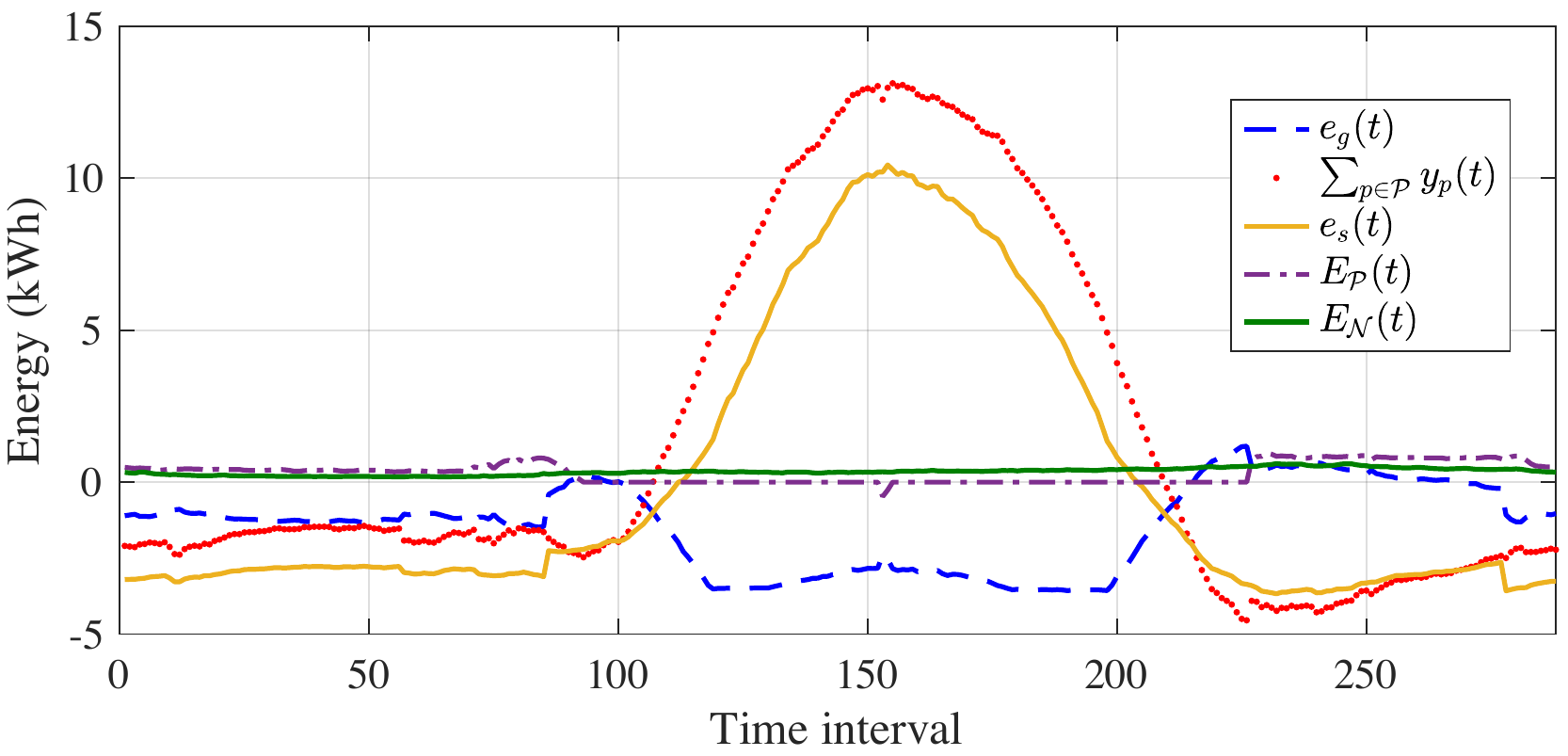}
\caption{Energy transactions in the energy trading system.}
\label{fig:enegy_transac.}
\end{figure}
Consequently, the CES system discharges $(e_s(t)<0)$, and the voltages at all buses remain within limits as shown in Fig.~\ref{fig:voltage_prof_together}(b). 
\begin{figure}[b!]
\centering
\includegraphics[width=0.75\columnwidth]{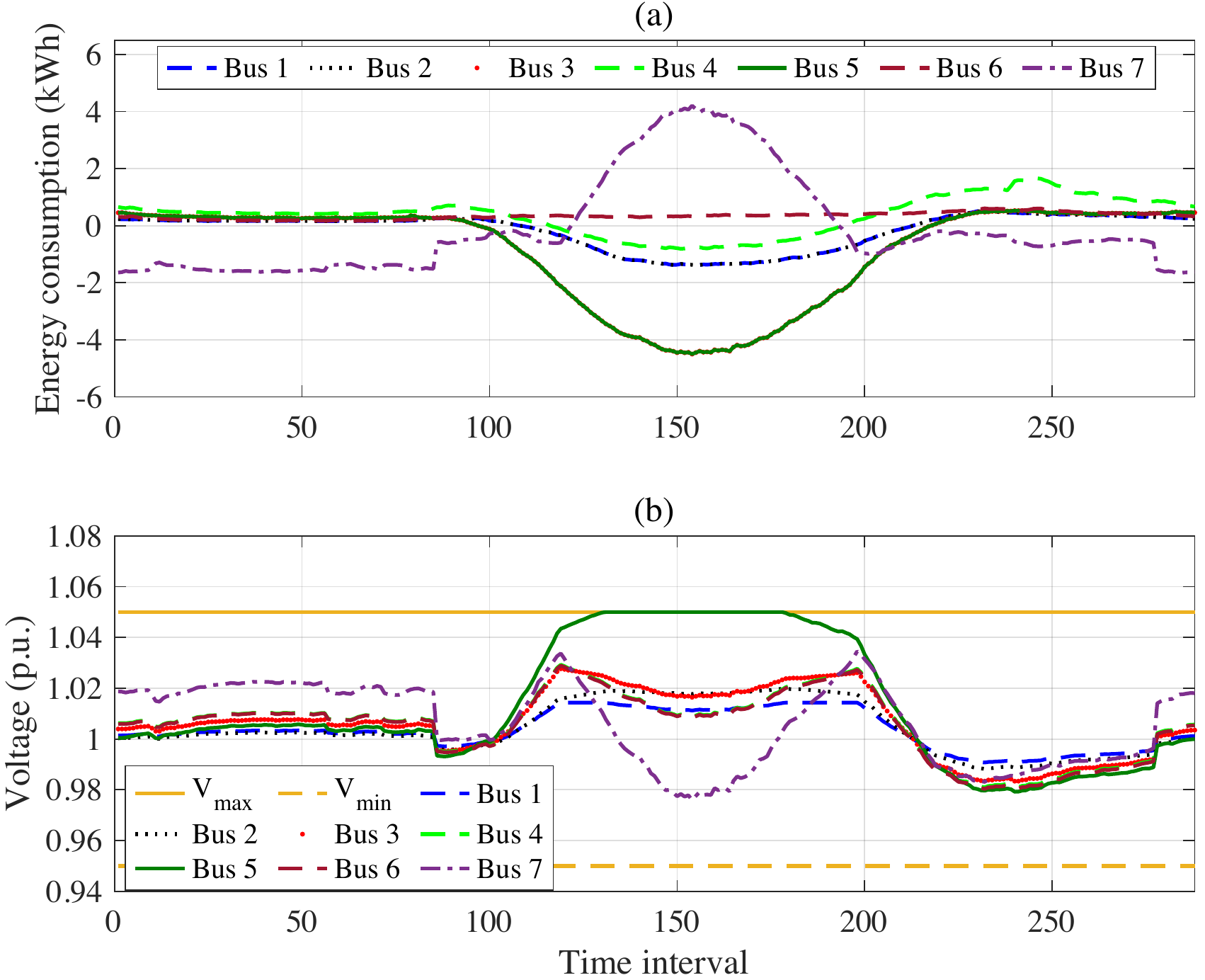}
\caption{(a) Aggregate bus energy consumptions and (b) bus voltages with the energy trading system.}
\label{fig:voltage_prof_together}
\end{figure}
Additionally, the total afternoon peak demand on the external grid is reduced by nearly $64\%$ (from 5.69 kWh to 2.02 kWh) compared to the baseline as depicted in Fig.~\ref{fig:Agg_GridLoad}(a).

At midday, between times 106 and 208, PV energy is plentiful, and nearly all users in $\mathcal{P}$ have positive surplus energy. Hence, each bus experiences negative aggregate energy consumptions in the baseline expect bus 6 where only non-participating users exist (see Fig.~\ref{fig:enrgInj_buses_baseline}(a)). Consequently, buses 1 - 7 experience over-voltages in the baseline as shown in Fig.~\ref{fig:enrgInj_buses_baseline}(b). In particular, the highest voltage $1.165$ p.u. occurs at bus 7 with the greatest negative energy consumption.
In the energy trading system, between times 106 and 208, the users $\mathcal{P}$ sell nearly all of their positive surplus energy to the CES and that leads to the charging mode of the CES system $(e_s(t) > 0)$ (see Fig.~\ref{fig:enegy_transac.}). As a result, the energy consumption at bus 7, i.e., $-\sum_{p\in \mathcal{P}_7}s_p(t) + e_s(t)$, becomes positive between times 121 and 192  as depicted in Fig.~\ref{fig:voltage_prof_together}(a). This helps regulate the voltages at all buses below the threshold 1.05 p.u during the over-voltage period as shown in Fig.~\ref{fig:voltage_prof_together}(b). Note that, in Figs. \ref{fig:enrgInj_buses_baseline}(a) and \ref{fig:voltage_prof_together}(a), the energy consumption profiles at buses 3 and 5 are the same because there are no users at bus 3. Additionally, in Figs.~\ref{fig:enrgInj_buses_baseline}(b) and \ref{fig:voltage_prof_together}(b), the voltage profiles at buses 4 and 6 closely follow each other due to the negligible voltage drop between the two buses.

The reduced external grid load $E(t)$ before time 105 and after time 209 leads to reduced grid prices $\lambda_g(t)$ compared to the baseline as shown in Fig.~\ref{fig:Agg_GridLoad}(b). Consequently, the average cumulative daily energy cost of the users $\mathcal{P}$ is reduced by 83\% (from $160$ AU cents to $27$ AU cents).  Additionally, the CES provider receives a cumulative revenue of  $7376$ AU cents in the energy trading system. Due to the reduced grid price, the average daily cost for the users $\mathcal{N}$ is reduced by $17\%$ (from $723$ AU cents to $597$ AU cents). 

\subsection{Voltage Constraints on the Energy Trading System}\label{results_3}
Here, we investigate the impacts of the energy trading schedules determined with and without the voltage constraint \eqref{eq:id23_a} in the game-theoretic optimization in Section~\ref{sec_system}. The same simulation setup as in Section \ref{results_1} is used. 

As shown in Fig.~\ref{fig:without_vol_const}, without the voltage constraints in the optimization, the voltages at all buses remain within the limits except bus 5.
\begin{figure}[b!]
\centering
\includegraphics[width=0.75\columnwidth]{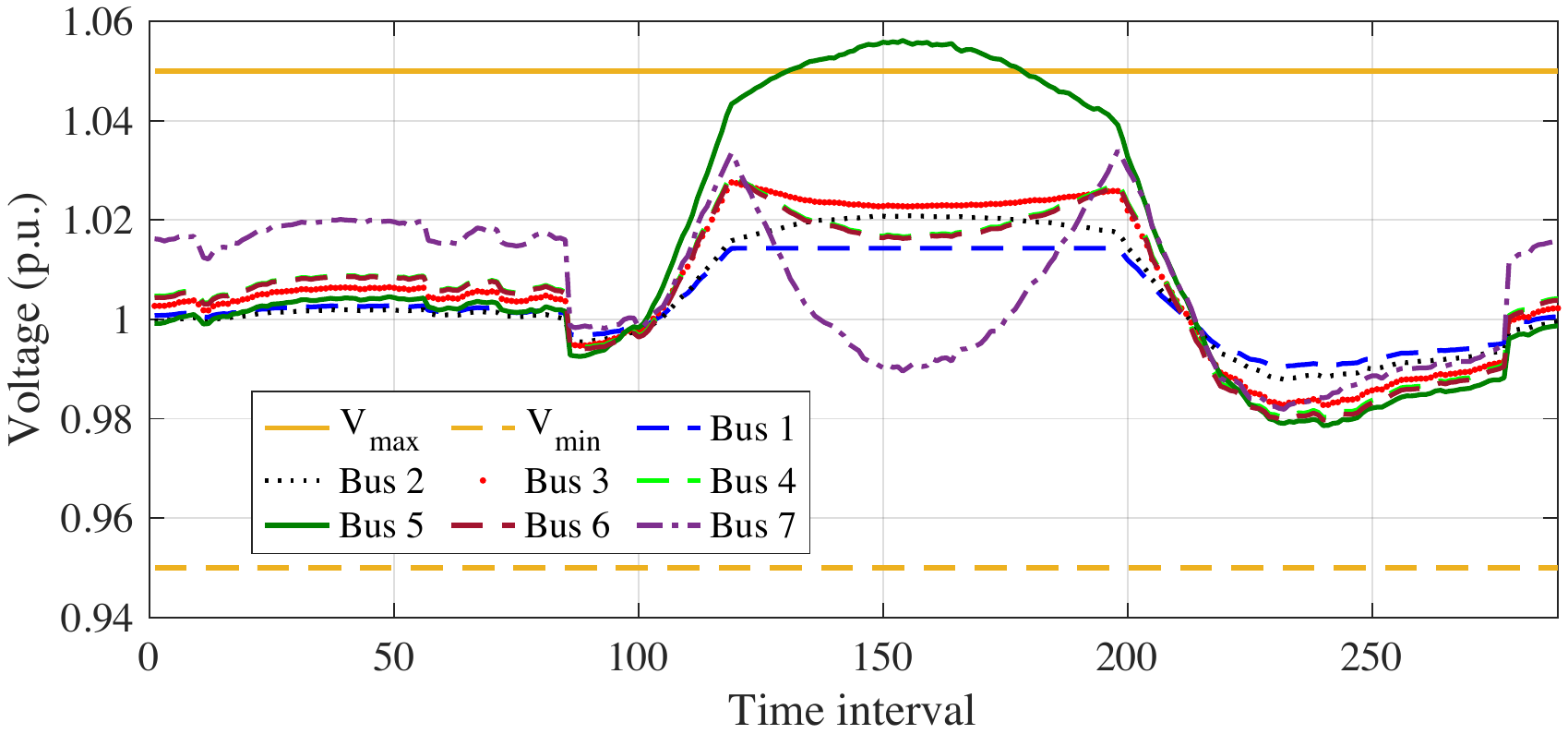}
\caption{Bus voltages without the voltage constraint in the optimization.}
\label{fig:without_vol_const}
\end{figure}
Note that in Fig. \ref{fig:without_vol_const}, the voltages at buses 4 and 6 closely follow each other due to the negligible voltage drop. Once the voltage constraints are introduced, the voltages at all buses fall within the limits as shown in Fig.~\ref{fig:voltage_prof_together}(b). With the voltage constraints, more energy is drawn by the CES system $(e_s(t) > 0)$ at midday than in the system without the voltage constraints and this, in turn, mitigates over-voltages of all buses including bus 5. The users $\mathcal{P}$ follow nearly the same energy trading strategies, $\sum_{p \in \mathcal{P}}y_p(t)$ and $E_{\mathcal{P}}(t)$, despite the voltage constraints. Additionally, the CES provider pays a higher price at midday for buying surplus PV energy from the users $\mathcal{P}$ when the voltage constraints are introduced. Hence, the average cumulative daily energy cost for the users $\mathcal{P}$ in the system with the voltage constraints (27 AUD cents) is reduced by $50\%$ compared to the system without the voltage constraints (52 AUD cents). However, the CES provider's revenue reduces by $9\%$ in the system with the voltage constraints (from 8091 AUD cents to 7376 AUD cents) as a result of paying a higher price for buying energy from the users $\mathcal{P}$ at midday.

\subsection{Comparison with a Centralized Energy Trading System}\label{centralized}
Here, we compare the performance of the game-theoretic energy trading system in Section~\ref{sec_system} with a centralized system. In the centralized system, the CES provider schedules the energy transactions $e_g(t)$ and $y_p(t)$ by minimizing the total cost paid by the entire community, i.e., the users $\mathcal{A}$ and the CES provider, to the retailer. Hence, its objective is to minimize $\sum_{t=1}^H \lambda_g(t)E(t)$ subject to constraints \eqref{eq:id2a},~\eqref{eq:id2}, \eqref{eq:id4}, \eqref{eq:id7}, \eqref{eq:id8aa}, and \eqref{eq:id23_a}. In the centralized system, the objective function does not include a price signal for the energy transactions $y_p(t)$ between the CES system and the users $\mathcal{P}$. The centralized system serves as a baseline and is a potential implementation for the energy trading between the CES system and the users $\mathcal{P}$. The simulation parameters are chosen as in Section~\ref{results_1}.

As shown in Fig.~\ref{fig:energt_tran_cent}, in the centralized system, the CES provider sells more energy to the external grid $(e_g(t)<0)$ before time 105 and after time 209 compared to the game-theoretic system.
\begin{figure}[b!]
\centering
\includegraphics[width=0.75\columnwidth]{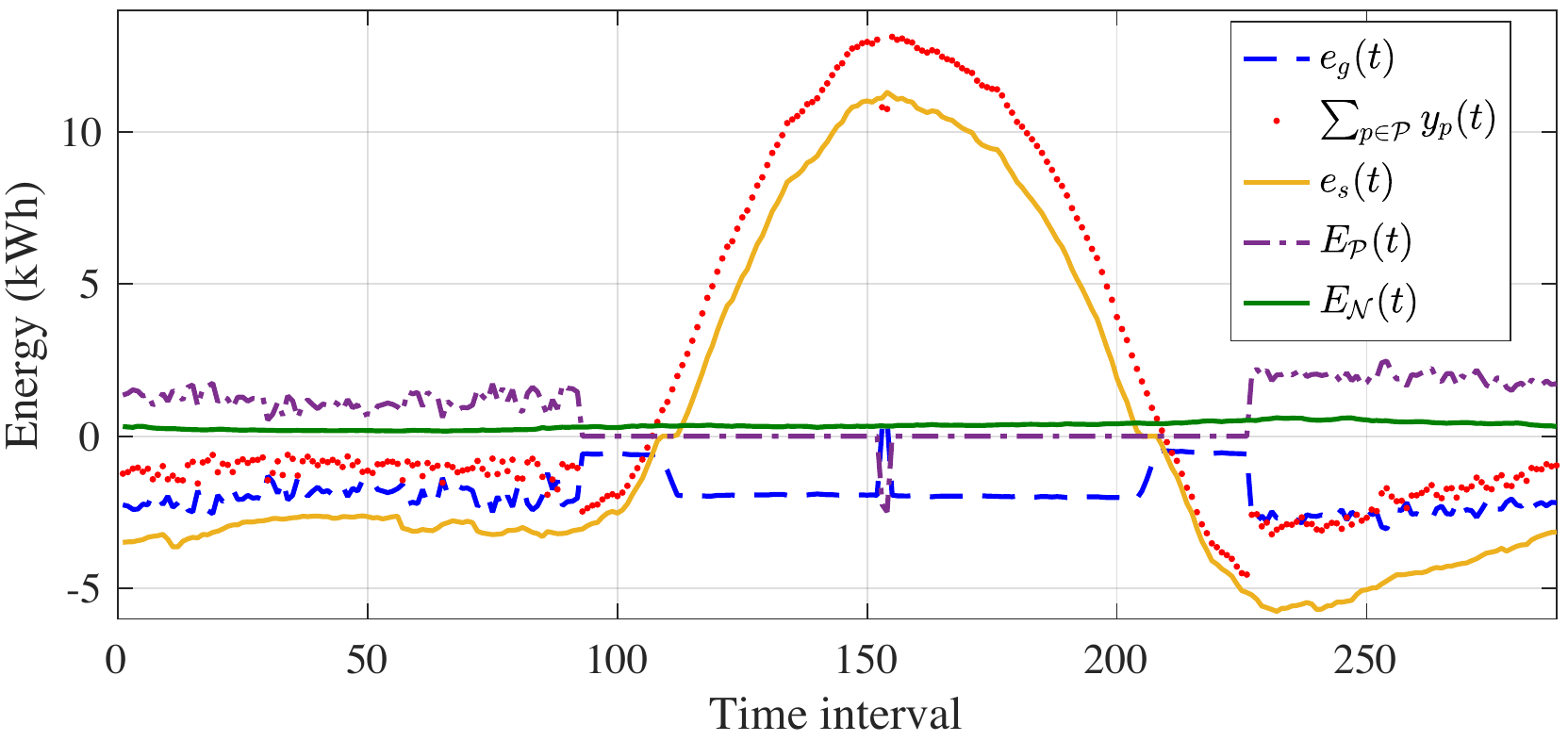}
\caption{Energy transactions in the centralized system.}
\label{fig:energt_tran_cent}
\end{figure}
This leads to a greater reduction of the total grid load $E(t) (= E_{\mathcal{P}}(t) + E_{\mathcal{N}}(t) + e_g(t))$ and hence, a greater reduction of the price $\lambda_g(t)$. At midday, the users $\mathcal{P}$ sell nearly all their positive surplus PV energy to the CES in both systems (see Fig.~\ref{fig:enegy_transac.} and Fig.~\ref{fig:energt_tran_cent}). However, the users $\mathcal{P}$ in the centralized system make nearly a zero income at midday with no price being offered to their CES energy transactions $y_p(t)$. With the CES price $\lambda_s(t)$, the users $\mathcal{P}$ in the game-theoretic system make a positive income by selling PV energy to the CES system at midday. Hence, in the centralized system, the average cumulative daily energy cost for the users $\mathcal{P}$ only reduces by 11\% whereas in the game-theoretic system, they receive nearly 83\% average cost reduction compared to the baseline without the CES system demonstrated in Section~\ref{results_1}. 

Similar to the game-theoretic system (as shown in Fig.~\ref{fig:voltage_prof_together}(b)), in the centralized system, the lowest voltage (0.96 p.u) occurs in bus 7 and the highest voltage (1.04 p.u.) occurs in bus 5 at midday. As shown in Fig.~\ref{fig:soc_central}, the CES system saturates at $B_{\text{max}}$ in the centralized system as a result of more energy drawn by the CES system ($e_s(t)>0$) at midday. 
\begin{figure}[t!]
\centering
\includegraphics[width=0.75\columnwidth]{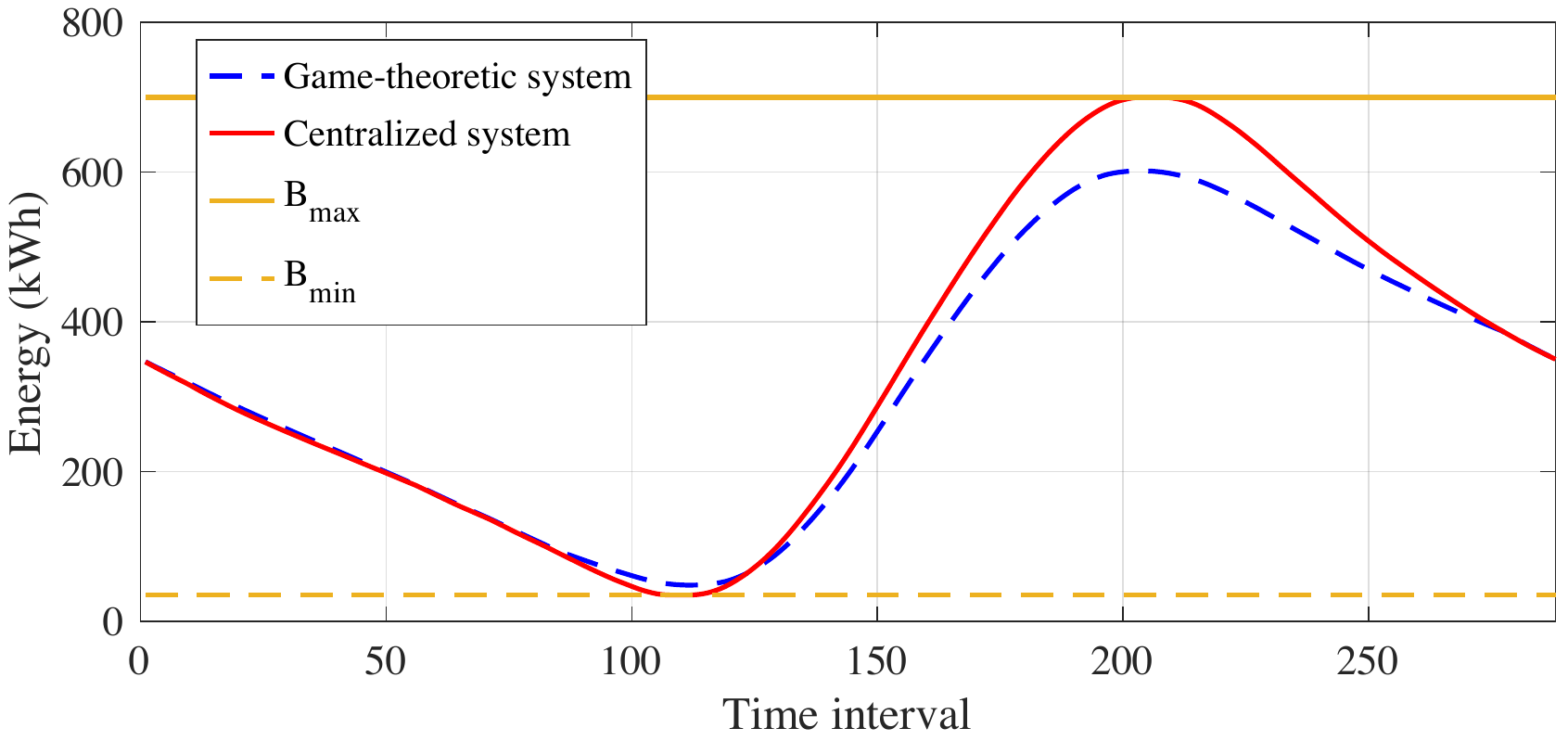}
\caption{Variations of the energy charge level of the CES system.}
\label{fig:soc_central}
\end{figure}
Consequently, the lowest voltage in the centralized system becomes 2\% less than that of the game-theoretic system (0.98 p.u.), and the highest voltage in the centralized system becomes 0.9\% less than that of the game-theoretic system (1.05 p.u.). However, Fig.~\ref{fig:soc_central} illustrates that the centralized system requires a larger energy storage capacity to this end.

\subsection{Impacts of Seasonal Changes in Demand and PV Power Profiles on the Energy Trading System}\label{result_4}
Here, the energy trading system performance is evaluated by changing the average daily PV power and demand profiles of the 55 users by season; summer, autumn, winter, and spring. The seasonal profiles were generated using the corresponding demand and PV power measurements in \cite{NextGen}. Here, $B_{\text{max}} = 950~\text{kWh}$,~$\gamma^{\text{ch}}_{\text{max}} = \gamma^{\text{dis}}_{\text{max}} = 300~\text{kW}$ so that the CES system can accommodate energy transactions without saturating at those limits when using the average demand and PV power profiles of all four seasons. All the other parameters are as in Section~\ref{results_1}. For comparison, the baseline without the CES system in Section~\ref{results_1} is used. To demonstrate the variations of seasonal aggregate bus energy consumptions, we include Fig.~\ref{fig:agg_energy_consump_seasons_bus7} that depicts the seasonal aggregate energy consumptions at bus 7 with the baseline. 
\begin{figure}[b!]
\centering
\includegraphics[width=0.75\columnwidth]{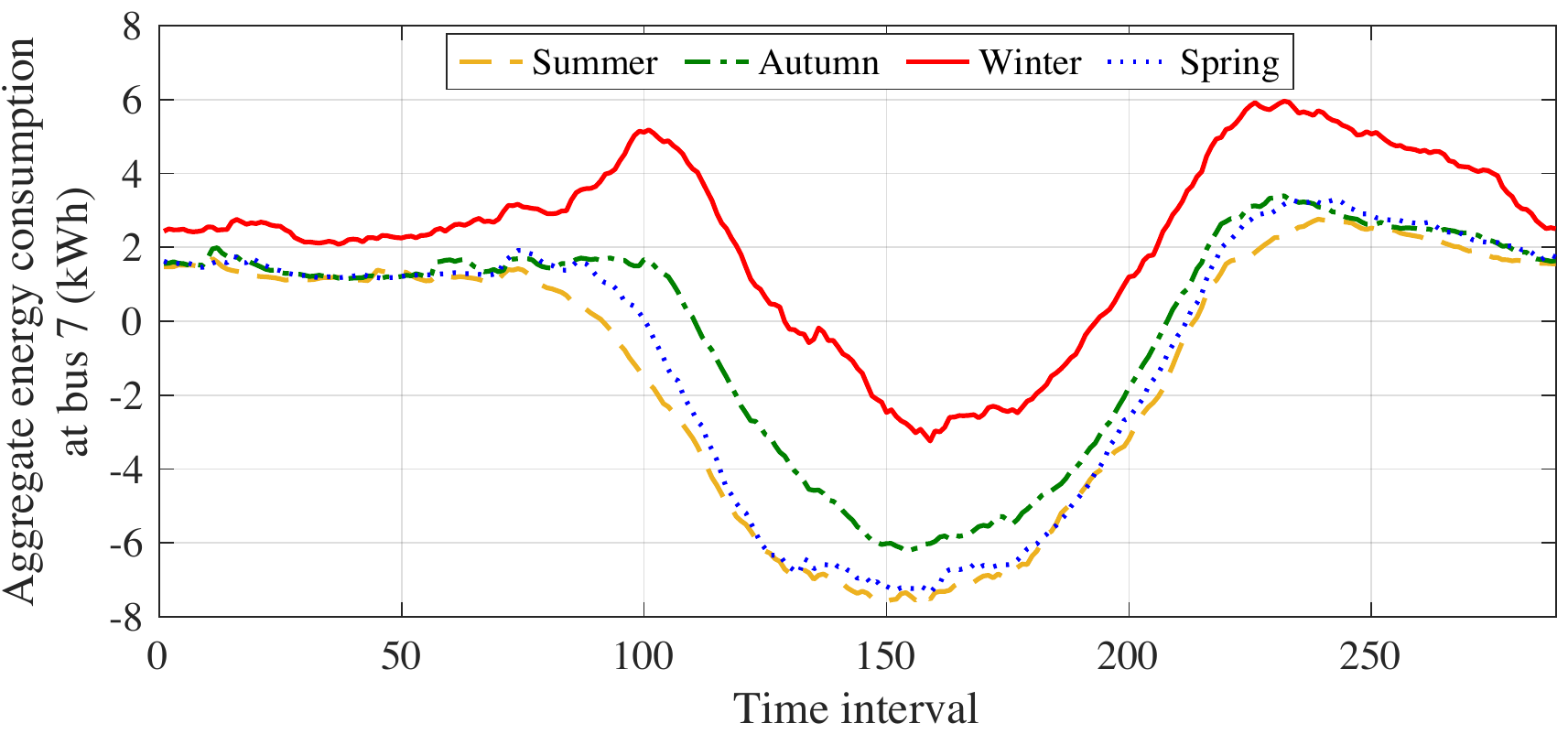}
\caption{Seasonal aggregate energy consumption at bus 7 in the baseline.}
\label{fig:agg_energy_consump_seasons_bus7}
\end{figure}

Fig.~\ref{fig:voltage_seasons} depicts the boxplots of 24-hr voltage distributions of all buses when the average demand and PV power profiles are changed by season.
\begin{figure}[t!]
\centering
\includegraphics[width=0.98\columnwidth]{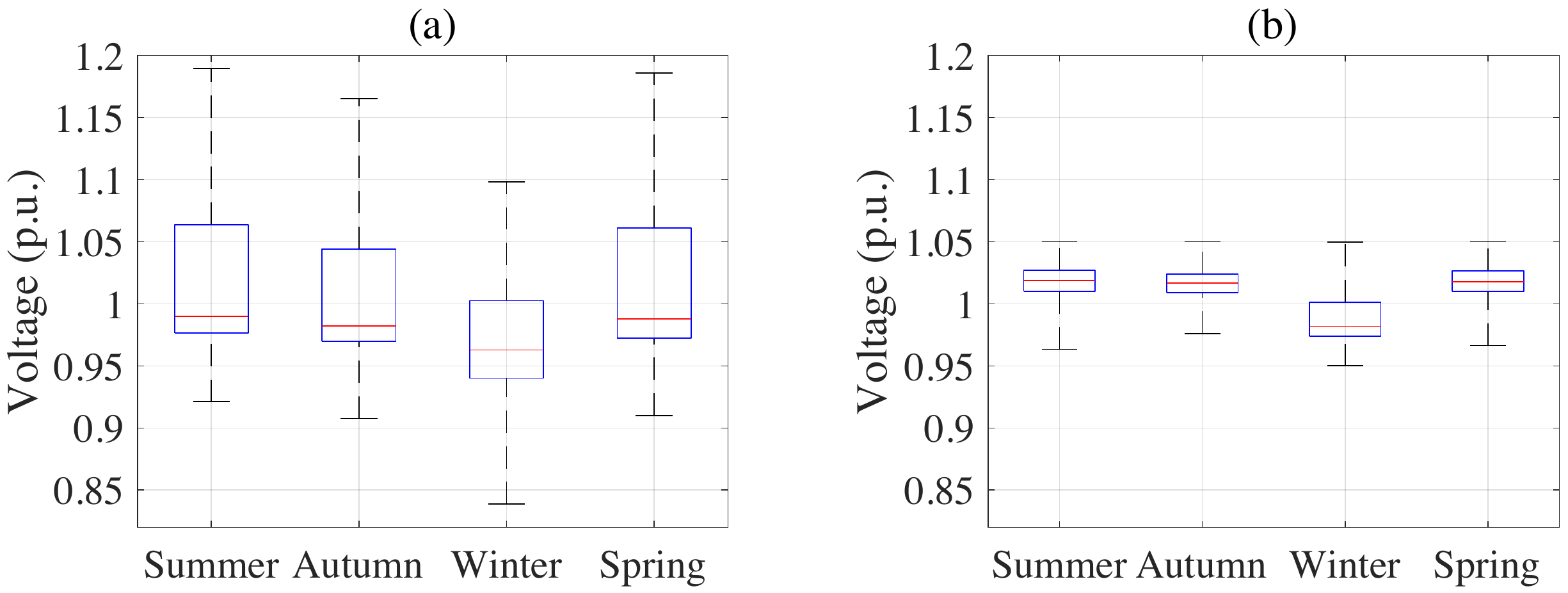}
\caption{Voltage distributions of all buses (a) in the baseline and (b) in the energy trading system with average demand and PV profiles of four seasons.}
\label{fig:voltage_seasons}
\end{figure}
The horizontal bars of the whiskers represent the maximum and the minimum voltages experienced by the feeder in the 24-hr period. As illustrated in Fig.~\ref{fig:voltage_seasons}(a), in the baseline, with summer profiles, the feeder experiences the highest over-voltage condition, 1.19 p.u. at bus 7, due to the greatest negative energy consumption at midday. Additionally, the highest under-voltage condition, 0.84 p.u. at bus 7, is experienced in winter with the greatest positive energy consumption by the users in the afternoon as shown in Fig.~\ref{fig:agg_energy_consump_seasons_bus7}. The energy trading system is capable of bringing the voltages within the limits for all seasons as shown in Fig.~\ref{fig:voltage_seasons}(b). 

Fig.~\ref{fig:CES_REV_PU_Cost_seasons} compares the normalized average cumulative energy cost of the users $\mathcal{P}$ and the normalized CES provider revenue in four seasons.
\begin{figure}[t!]
\centering
\includegraphics[width=0.98\columnwidth]{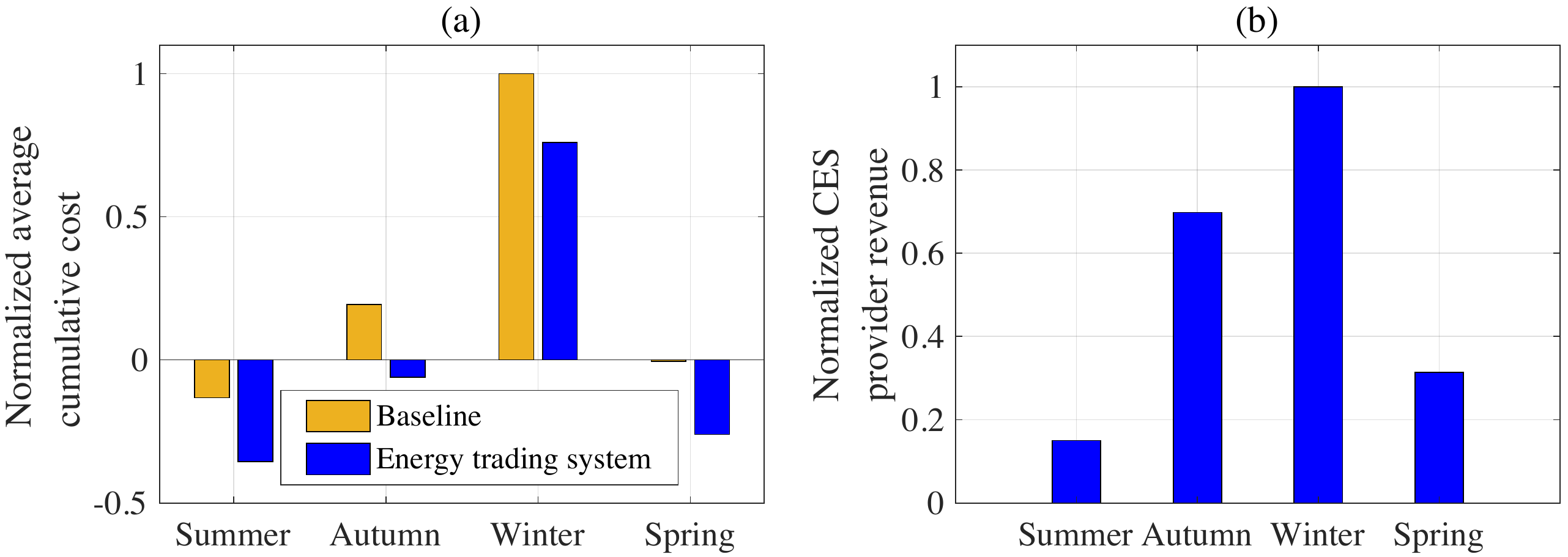}
\caption{ (a) Normalized average cumulative cost of the users $\mathcal{P}$ (b) normalized CES provider revenue with average demand and PV profiles of four seasons.}
\label{fig:CES_REV_PU_Cost_seasons}
\end{figure}
As shown in Fig.~\ref{fig:CES_REV_PU_Cost_seasons}(a), in winter, the users $\mathcal{P}$ incur the highest positive energy cost in both baseline and the proposed energy trading systems. This is because, in winter, the users buy more energy than selling due to having the greatest energy demand with less PV power generation. In summer, due to having plenty of excessive PV power, the users $\mathcal{P}$ generate the highest revenue (negative cost) by selling PV power, and the users $\mathcal{P}$ receive more revenue in the energy trading system than in the baseline. In Fig.~\ref{fig:CES_REV_PU_Cost_seasons}(b), the CES provider's revenue is presented only for the energy trading system as there is no CES system in the baseline. 
Because the users $\mathcal{P}$ buy more energy from the CES system in winter, the CES provider receives the greatest revenue in winter, whereas, in summer, due to selling more PV energy by the users $\mathcal{P}$ to the CES system, the CES provider receives the least revenue.

\section{Conclusion}\label{conclusion}
With the ability to locate close to users, community energy storage (CES) systems can be utilized to deliver demand-side management and voltage support for low-voltage distribution networks. In this paper, we have investigated the extent to which a CES system can reduce voltage excursions and peak energy demand of a radial distribution network by developing a decentralized energy trading system between a CES system and the users with rooftop photovoltaic (PV) power generation. By employing a linearized branch flow model, a voltage-constrained Stackelberg game was developed where the CES provider and users can maximize their personal economic benefits. It has been shown that the energy trading system can deliver significant peak energy demand reduction and economic benefits for both the CES provider and the users while satisfying the network voltage limits.

Future work includes developing a stochastic model to incorporate uncertainties of demand and PV power generation and imperfect knowledge from players' actions in the energy trading system, exploring the energy trading system operation to accommodate the possibility of unbalanced conditions in three-phase distribution networks, and extending the system to incorporate the PV and CES inverter reactive power control.






\begin{IEEEbiography}{Chathurika P. Mediwaththe}
(S’12-M’17) received
the B.Sc. degree (Hons.) in electrical and electronic
engineering from the University of Peradeniya, Sri
Lanka. She completed the PhD degree in electrical
engineering at the University of New South Wales,
Sydney, Australia in 2017. From 2013-2017,
she was a researcher with Data61-CSIRO
(previously NICTA), Sydney, NSW, Australia. She
is currently a research fellow with the Research
School of Electrical, Energy and Materials Engineering and the Battery Storage and Grid
Integration Program, Australian National University, Australia. She is the chair of the Young Professionals group of the IEEE Australian Capital Territory chapter. Her current
research interests include electricity demand-side management, renewable
energy generation and energy storage integration, microgrids and distribution network operation, game theory and optimization 
for distributed energy resource allocation.
\end{IEEEbiography}

\begin{IEEEbiography}{Lachlan Blackhall}
(S’07-M’11-SM’17) received the BE (Hons 1M) and BSc (Adv. Mathematics) from the University of Sydney, Australia in 2007. He received a PhD in Systems and Control Theory at the Australian National University (ANU), Australia in 2011. He is co-founder and previously CTO of Reposit Power and is currently Associate Professor and Head, Battery Storage and Grid Integration Program at the ANU. He is a fellow of both the Australian Academy of Technology and Engineering (ATSE) and of the Institute of Engineers Australia (IEAust). His current research interests include the optimisation and control of distributed energy resources in electricity distribution networks and the operation of power systems with high uptake renewable generation and distributed energy resources.
\end{IEEEbiography}

\ifCLASSOPTIONcaptionsoff
  \newpage
\fi
\end{document}